%% file: main.tex
\documentclass[sigconf,authorversion,nonacm]{acmart}
\AtBeginDocument{%
  }

\setcopyright{acmlicensed}
\copyrightyear{2018}
\acmYear{2018}
\acmDOI{XXXXXXX.XXXXXXX}
\acmConference[Conference acronym 'XX]{Make sure to enter the correct
  conference title from your rights confirmation email}{June 03--05,
  2018}{Woodstock, NY}
\acmISBN{978-1-4503-XXXX-X/2018/06}

\usepackage{xcolor}

\usepackage[linesnumbered,ruled]{algorithm2e}
\makeatletter
\let\oldnl\nl
\newcommand{\nonl}{\renewcommand{\nl}{\let\nl\oldnl}}
\makeatother

\SetCommentSty{mycommfont}

\usepackage{amsmath,amsthm,graphicx, algpseudocode}
\usepackage{booktabs}   
\usepackage{tabularx}   
\usepackage{graphicx}  

\usepackage{enumitem}

\SetAlFnt{\small}
\SetKwBlock{KwFunction}{function}{}
\SetKwComment{Comment}{$\triangleright$ }{}
\SetKwInput{KwOn}{On}

\usepackage{listings}

\usepackage[font=small,skip=2pt]{caption}

\usepackage{tikz}
\usetikzlibrary{arrows.meta, positioning, shapes.geometric}

\input{my_macros.tex}

\newtheorem{definition}{Definition}
\newtheorem{theorem}{Theorem}
\newtheorem{claim}{Claim}
\newtheorem{proposition}{Proposition}
\newtheorem{lemma}{Lemma}
\newtheorem{corollary}{Corollary}

\newtheorem{note2}{Note}

\title{Resilient Alerting Protocols for Blockchains}
\thanks{This paper was accepted for publication at CCS26}

\author{Marwa Mouallem}
\email{marwamouallem@cs.technion.ac.il}
\affiliation{%
  \institution{Technion, Chainlink Labs, IC3}
  \country{}
}

\author{Lorenz Breidenbach}
\email{lorenz@chainlinklabs.com}
\affiliation{%
  \institution{Chainlink Labs}
  \country{}
}

\author{Ittay Eyal}
\authornote{Part of the work was done while being an advisor to Chainlink Labs.}
\email{ittay@technion.ac.il}
\affiliation{%
  \institution{Technion, IC3}
  \country{}
}

\author{Ari Juels}
\email{juels@cornell.edu}
\affiliation{%
  \institution{ Cornell Tech, Chainlink Labs, IC3}
  \country{}
}

\date{}

\begin{document}

\begin{abstract}
Smart contracts are stateful programs deployed on blockchains; they secure over a trillion dollars in transaction value per year. 
High-stakes smart contracts often rely on timely alerts about external events, but prior work has not analyzed their resilience to an attacker suppressing alerts via bribery.
We formalize this challenge in a cryptoeconomic setting as the \emph{alerting problem}, giving rise to a game between a bribing adversary and~$n$ rational participants, who pay a penalty if they are caught deviating from the protocol.
We establish a quadratic, i.e.,~$O(n^2)$, upper bound, whereas a straightforward alerting protocol only achieves~$O(n)$ bribery cost. 

We present a \emph{simultaneous game} that asymptotically achieves the quadratic upper bound and thus asymptotically-optimal bribery resistance. 
We then present two protocols that implement our simultaneous game:
The first leverages a strong network synchrony assumption. 
The second relaxes this strong assumption and instead takes advantage of trusted hardware and blockchain proof-of-publication to establish a timed commitment scheme. 
These two protocols are constant-time but incur a linear storage overhead on the blockchain. 
We analyze a third, \emph{sequential alerting} protocol that optimistically incurs no on-chain storage overhead, at the expense of~$O(n)$ worst-case execution time. 
All three protocols achieve asymptotically-optimal bribery costs, but with different resource and performance tradeoffs.
Together, they illuminate a rich design space for practical solutions to the alerting problem.

\end{abstract}




\settopmatter{printacmref=false}

\maketitle

\input{in_progress.tex}

\begin{acks}
This paper was edited for grammar using ChatGPT.
\end{acks}

\section*{Ethical considerations}
This work studies the alerting problem in blockchain systems at the level of abstract models and incentive mechanisms. It does not evaluate, attack, or disclose vulnerabilities in any specific deployed protocol or implementation.

\paragraph{Stakeholders.}
Relevant stakeholders include protocol designers, smart contract developers, and system operators who rely on off-chain alerting mechanisms, as well as end users whose assets depend on the correct and timely functioning of such systems. These stakeholders benefit from a clearer understanding of the economic limits of alerting mechanisms and from protocol designs with provable resistance to bribery-based suppression.

\paragraph{Potential risks and misuse considerations.}
Our results characterize fundamental bounds on bribery resistance under explicit modeling assumptions. Our analysis is purely theoretical and does not provide system-specific parameters or attack implementations, and all results are presented alongside defensive constructions that achieve provably optimal resistance within the same model.

\paragraph{Risk mitigation and justification.}
We believe that making these limitations explicit strengthens, rather than weakens, the security of real systems by preventing overconfidence in alerting mechanisms that do not meet the necessary economic criteria. By focusing on abstract models and provably robust designs, the paper aims to guide protocol designers toward safer constructions rather than enable practical exploitation.

\paragraph{Human subjects and data.}
This research does not involve human subjects, personal data, or experimentation on live systems. No IRB approval was required.

\section*{Open Science}
The results in this paper do not rely on experimental data or code.
All proofs and analyses are included in the paper and its appendices.

\bibliographystyle{ACM-Reference-Format}
\bibliography{main}

\appendix

\section{Conditional Bribes}
\label{app:conditional_bribes}
In this section, we consider an alternative adversary strategy where the bribe is conditional on the success of the attack.
The game structure is the same as in section~\ref{sec:simultaneous_alerting_game}.
The only difference is the adversary's bribe strategy.

First, the adversary offers a bribe vector $(\beta_1,\dots,\beta_n)$ specifying how much each node would be paid for not alerting if no alert is raised.
Otherwise, nodes are not paid a bribe.
The bribes are private information between the adversary and each node, meaning that no other node knows the bribe offered to a specific node. 
Then, given these offers, each node $i \in \mathcal{N}$ decides whether to accept the bribe and skip alerting or to alert and not get the bribe, choosing an action $a_i \in \{\mathsf{Alert}, \noalert\}$. 
At the end of the round, the adversary pays the bribe only if all nodes choose \noalert.
Nodes that follow the protocol and alert when they should get rewarded.

Similar to the original analysis, denote by~$\alerters$ the set of players who chose $\alert$.
If~$\alerters$ is not empty, all nodes~$j \notin \alerters$ pay a penalty of~$\penaltySlash$
and all nodes~$i \in \alerters$, share the slashed value equally as a reward.
Each alerter gets~$\alertReward{i} = \frac{ \penaltySlash \cdot (n - |\alerters|) }{|\alerters|}$.
Node~$i$ may choose not to alert in exchange for a bribe of~$\beta_i$ conditioned on the attack's success.

In summary, the utility of each node~$i$ is her payoff:
$$
\label{eq:sim_payoff2}
\rho_i(a_i,\beta_i,\alerters)
= 
\begin{cases}
  \begin{array}{ll}
    \dfrac{\penaltySlash \cdot (n-|\alerters|)}{|\alerters|} 
      & \text{if } a_i = \mathsf{Alert}\\  

     - \penaltySlash 
      & \text{if } a_i \neq \mathsf{Alert} \wedge \alerters \neq \emptyset \\

    \beta_i 
      & \text{if } \alerters = \emptyset \\ 
  \end{array}
\end{cases}
$$

And for the adversary, if the attack is successful, i.e., no alert is raised, she gains~$\advGain$ but pays all bribes.
If an alert is raised, she gains nothing but pays the bribes to nodes that did not alert.
Thus, the adversary's utility is her payoff:
\[
\label{eq:sim_payoff_adv2}
\rho_{\text{adv}}(\beta_1,\dots,\beta_n,\alerters) = 
\begin{cases}
\advGain - \sum\limits_{i \in \mathcal{N}} \beta_i, & \alerters = \emptyset\\

0, & \alerters \neq \emptyset \\
\end{cases}
\]

We now show that if the adversary's gain from a successful attack is~$\advGain < \penaltySlash \, n (n-1)$, there exists at least one node that will alert.

\begin{claim}[\alert\ is dominant when $\beta_i  < \penaltySlash (n-1)$]
\label{clm:alert-dominant2}
Consider the second step of the simultaneous alerting game assuming conditional bribes.
If node~$i$'s offered bribe satisfies~$\beta_i< \penaltySlash (n-1)$, then the action $\alert$ dominates $\noalert$. 
\end{claim}

\begin{proof}
Consider a node~$i$ and let~$y\in\{0,1,\ldots,n-1\}$ be the number of other alerters.
We consider two cases based on the value of~$y$.
\begin{enumerate}
  \item \textbf{$y=0$.} If nobody else alerts, then~$u_i(\noalert,\beta_i,\alerters)=\beta_i$ and~$u_i(\alert,\beta_i,\alerters)=\penaltySlash(n-1)$.
Since $\beta_i < \penaltySlash (n-1)$, choosing $\alert$ is better.

\item \textbf{$y\ge 1$.} If at least one other node alerts, then~$u_i(\noalert,\beta_i,\alerters)=-\penaltySlash<0$ and~${u_i(\alert,\beta_i,\alerters)=\penaltySlash\frac{n-1-y}{y+1}\ge 0}$.
Thus $\alert$ is better.
\end{enumerate}

Since $\alert$ yields a higher utility in every case, it dominates $\noalert$ for all players.
\end{proof}

We showed that if the offered bribe~$\beta_i\leq\penaltySlash$, then the dominant strategy for the node~$i$ is to choose $\alert$.
Therefore, if the adversary's gain from a successful attack is~$\advGain < \penaltySlash \, n (n-1)$, there exists at least one node that will alert.

\section{Early Reveal Attacks}
\label{app:early_reveal_attacks}
In this section, we sketch an example of an attack that an adversary can launch against a simultaneous alerting protocol if nodes can reveal their actions before other nodes have committed to theirs.
We show the attack costs~$O(n \log n)$. 
Consider the following adversarial strategy.
First, the adversary chooses an arbitrary ordering of the nodes.
For simplicity, assume the ordering is by node indices, i.e., node~$1$ goes first, then node~$2$, and so on.
The adversary offers bribes to one node at a time, starting from node~$1$.

The adversary offers the first node a bribe of~$\beta_1 = \penaltySlash+ \varepsilon > 0$, in exchange for choosing \noalert and revealing its action immediately. 
If node~$1$ declines, the adversary stops the attack, in which case all nodes will choose \alert\ by Claim~\ref{clm:alert-dominant}, and node~$1$'s utility will be~$0$ since no other node is slashed.
If, however, node~$1$ accepts the bribe, its utility becomes~$\varepsilon$, which is strictly better. 
Therefore, it chooses not to alert, it reveals its action immediately, and the adversary proceeds to node~$2$.

The adversary proves to node~$2$ that node~$1$ chose \noalert, and offers node~$2$ a bribe of~$\beta_2 = \frac{\penaltySlash}{n-1} + \penaltySlash + \varepsilon$, in exchange for choosing \noalert\ and revealing its action immediately.
If node~$2$ declines the bribe, then by Claim~\ref{clm:alert-dominant}, all nodes with higher indices will choose \alert, and its utility will be~$\frac{\penaltySlash}{n-1}$.
If, however, node~$2$ accepts the bribe, its utility is~$\beta_2 = \frac{\penaltySlash}{n-1} + \penaltySlash + \varepsilon$, which is strictly better.
Thus, it chooses not to alert, it reveals its action immediately, and the adversary proceeds to node~$3$.

The adversary continues this process, for each node~$i$, she proves that all previous~$i-1$ nodes chose \noalert and offers node~$i$ a bribe of~$\beta_i = \penaltySlash \frac{i-1}{n-i+1} + \penaltySlash+ \varepsilon$.
If node~$i$ declines the bribe, then by Claim~\ref{clm:alert-dominant}, all nodes with higher indices will choose \alert, and its utility will be~$\penaltySlash \frac{i-1}{n-i+1}$.
If, however, node~$i$ accepts the bribe, its utility is~$\beta_i = \penaltySlash \frac{i-1}{n-i+1} + \penaltySlash + \varepsilon$, which is strictly better.
Thus, it chooses not to alert, it reveals its action immediately, and the adversary proceeds to node~$i+1$.

This process continues until the last node~$n$, for which the adversary proves all previous nodes chose \noalert, and offers a bribe of~$\beta_n = \penaltySlash (n-1) + \varepsilon$, which is strictly higher than its utility if it were to alert, which equals~$\penaltySlash(n-1)$, so it chooses not to alert.

The total bribe in this attack is:
\[
\sum_{i=1}^{n} \beta_i = \sum_{i=1}^{n} \left( \penaltySlash \frac{i-1}{n-i+1} + \penaltySlash + \varepsilon \right) = O(n \log n)
\]

Overall, the adversary suppresses all alerts for a cost~$O(n \log n)$ by breaking the simultaneity assumption.

\section{Simultaneous Alerting Game With Non-Zero Operator Cost}
\label{app:non_zero_operator_cost}
We present a \emph{simultaneous alerting game} and show it achieves quadratic bribery resistance.
We start by defining the game as a Stackelberg game~(\S\ref{sec:simultaneous_alerting_game_opcost}).
We characterize its equilibria under different ranges of the adversary's gain from a successful attack.
First, we identify conditions under which either alerting or not alerting is a dominant strategy, leading to a symmetric pure equilibrium in each case~(\S\ref{sec:pure-strategy-nash-equilibria_opcost}).
We then show that in the intermediate range of bribes no symmetric pure equilibrium exists~(\S\ref{sec:pure-strategy-nash-equilibria_opcost}). 
Finally, we extend the analysis to mixed-strategy equilibria and study the adversary's optimal bribe choice and expected utility~(\S\ref{sec:mixed-strategy-equilibria_opcost}).

\subsection{Game Definition}
\label{sec:simultaneous_alerting_game_opcost}
We provide the full game analysis considering an operator cost of $\operatorCost \geq 0$.
Again, two types of players are involved: the set of nodes $\mathcal{N}$ and an adversary.
The game proceeds in two stages.
First, the adversary offers bribes by choosing a vector $(\beta_1,\dots,\beta_n)$ specifying how much each node would be paid for not alerting, regardless of the game outcome.

The adversary pays the bribe to any node that chooses \noalert, regardless of the game's final outcome. 
The bribes are private information between the adversary and each node, meaning that no other node knows the bribe offered to another node.
Then, given these offers, each node $i \in \mathcal{N}$ decides whether to accept the bribe and skip alerting or to alert and not get the bribe, choosing an action $a_i \in \{\alert, \noalert\}$.
Afterwards, the adversary pays the bribe only to nodes whose on-chain behavior matches the \noalert action.
The protocol rewards nodes that follow the protocol.

In an execution of the game, denote by~$\alerters$ the set of players who choose $\alert$.
If~$\alerters$ is not empty, all nodes~$j \notin \alerters$ pay a penalty of~$\penaltySlash$
and all nodes~$i \in \alerters$ share the slashed value equally as a reward.
Each alerter gets~$\frac{ \penaltySlash \cdot (n - |\alerters|) +\operatorCost }{|\alerters|}$.
Node~$i$ may choose not to alert in exchange for a bribe of~$\beta_i$.
A node that alerts receives no bribe.

In summary, the utility of each node~$i$ is its payoff:
\begin{equation}
\label{eq:sim_payoff_opcost}
\rho_i(a_i,\beta_i,\alerters)
= 
\begin{cases}
    \frac{\penaltySlash \cdot (n-|\alerters|) +\operatorCost }{|\alerters|}, 
      &  a_i = \alert\\  

    \beta_i - \penaltySlash, 
      &  a_i = \noalert, \alerters \neq \emptyset \\

    \beta_i, 
      & \alerters = \emptyset .
\end{cases}
\end{equation}

For the adversary, if none of the nodes alert, she gains~$\advGain$ but pays all bribes.
If a node alerts, the adversary gains nothing but pays the bribes to nodes that did not alert.
Thus, the adversary's utility is her payoff:
\begin{equation}
\label{eq:sim_payoff_adv_opcost}
\rho_{\text{adv}}(\beta_1,\dots,\beta_n,\alerters) = 
\begin{cases}
\advGain - \sum\limits_{i \in \mathcal{N}} \beta_i, & \alerters = \emptyset\\

- \sum\limits_{i \in \mathcal{N} \setminus \alerters} \beta_i, & \alerters \neq \emptyset 
\end{cases}
\end{equation}

\subsection{Pure Strategy Nash Equilibria}
\label{sec:pure-strategy-nash-equilibria_opcost}
We analyze the pure-strategy Nash equilibria of the simultaneous alerting game.
We show that there are two cases where a pure dominant strategy exists.
If the bribe is sufficiently high, then each node's dominant strategy is to choose $\noalert$.
And if the bribe is sufficiently low, then each node's dominant strategy is to choose $\alert$.
In the latter case, the attacker's optimal strategy is to not offer any bribe.

\begin{claim}[{\noalert} dominance]
\label{clm:noalert-dominant_opcost}
In the simultaneous alerting game, if the adversary's offered bribe for a node~$i$ is~$\beta_i\ge \penaltySlash(n-1) + \operatorCost$, then the action~$\noalert$ weakly dominates $\alert$ for node~$i$. 
And $\noalert$ is strictly dominant if $\beta_i>\penaltySlash(n-1) + \operatorCost$.
\end{claim}

\begin{proof}
Consider a node~$i$ and let~${y\in\{0,1,\ldots,n-1\}}$ denote the number of other alerters in an execution of the game.
If~$i$ alerts while~$y$ others alert, i.e.,~${|\alerters|=y+1}$, its payoff is~$\frac{\penaltySlash(n-y-1) + \operatorCost}{y+1}$.
Otherwise, if~$i$ does not alert and someone alerts, its payoff is $\beta_i-\penaltySlash$.
Finally, if no one alerts i.e.,~$y=0$, choosing $\noalert$ earns node~$i$ the offered bribe~$\beta_i$ while a single alerter earns $\penaltySlash(n-1) + \operatorCost$.

We divide the proof into two cases based on the value of $y$.

\begin{enumerate}
\item \textbf{$y=0$.}
Comparing between $\noalert$ and $\alert$ gives ${u_i(\noalert,\beta_i, \alerters=\varnothing)=\beta_i}$ and ${u_i(\alert, \beta_i, \alerters=\{i\})=\penaltySlash(n-1) + \operatorCost}$.
By assumption $\beta_i\ge \penaltySlash(n-1) + \operatorCost$, so $\noalert$ is at least as good, and strictly better if~${\beta_i>\penaltySlash(n-1) + \operatorCost}$.

\item \textbf{ $y\ge 1$.}
Again we compare the utility of the two actions.
$u_i(\noalert, \beta_i, \alerters)=\beta_i-\penaltySlash$, while~$u_i(\alert, \beta_i, \alerters)=\frac{\penaltySlash(n-1-y) + \operatorCost}{y+1}$.
The function~$f(y):=\frac{\penaltySlash(n-y-1) + \operatorCost}{y+1}$ is decreasing in~$y$ for
$y\ge 1$, so its maximum over $\{1,\ldots,n-1\}$ is at $y=1$ with~${f(1)=\frac{\penaltySlash(n-2) + \operatorCost}{2}}$. 
Thus it suffices to show~$
\beta_i-\penaltySlash \;\ge\; \frac{\penaltySlash(n-2) + \operatorCost}{2}.
$

This holds whenever $\beta_i\ge \frac{\penaltySlash \, n +\operatorCost}{2}$, and in particular under our assumption $\beta_i\ge \penaltySlash (n-1) + \operatorCost$ because ${\penaltySlash(n-1) + \operatorCost \ge \frac{\penaltySlash \, n + \operatorCost}{2}}$ for all $n\ge 2$. 
Therefore, for every~${y\ge 1}$,~${u_i(\noalert, \alerters)\ge u_i(\alert, \alerters)}$, with strict inequality when ${\beta_i>\penaltySlash(n-1)+ \operatorCost}$.
\end{enumerate}
Combining the two cases, $\noalert$ dominates $\alert$ for all nodes~$i$ whose~$\beta_i\ge \penaltySlash(n-1) + \operatorCost$, and is strictly dominant when $\beta_i>\penaltySlash(n-1) + \operatorCost$.
\end{proof}

\begin{corollary}[Profitable bribery when~$\advGain > \penaltySlash \, n(n-1)+ n \, \operatorCost$]
\label{cor:profitable-bribery_opcost}
In the full two-stage alerting game, if the adversary's gain from a successful attack is~${\advGain > \penaltySlash \, n(n-1) + n \, \operatorCost}$, then there exists a profitable bribery strategy for the adversary that induces a unique pure-strategy equilibrium in the second stage of the game where all nodes choose $\noalert$.
\end{corollary}
\begin{proof}
By Claim~\ref{clm:noalert-dominant_opcost}, for every node $i$, if~${\beta_i > \penaltySlash \, (n-1) + \operatorCost}$, then $\noalert$ dominates $\alert$.
Hence, if the adversary offers $\beta_i > \penaltySlash(n-1) + \operatorCost$ to all nodes, the nodes' unique best-response profile is that all nodes choose $\noalert$. 
The adversary's utility at these $\beta_i$ values and nodes' response is $u_{\mathrm{adv}} = \advGain - \sum_{i=1}^{n} \beta_i$ if she chooses to bribe the nodes, and $u_{\mathrm{adv}} = 0$ if she does not. 
Thus, when the adversary's gain from a successful attack is~${\advGain > \penaltySlash \, n(n-1) + n \, \operatorCost}$, then bribery is profitable.
\end{proof}

We conclude that when the briber's gain from a successful attack is over~$\penaltySlash n(n-1) + n \, \operatorCost$, she can ensure that all nodes play $\noalert$ and achieve a positive utility.
We now consider the opposite case where the bribe is too low.

\begin{claim}[\alert\ is dominant when $\beta_i  \leq \penaltySlash + \frac{\operatorCost}{n}$]
\label{clm:alert-dominant_opcost}
In the simultaneous alerting game, if the adversary's offered bribe for a node~$i$ is~$\beta_i \leq \penaltySlash + \frac{\operatorCost}{n}$, then the action~$\alert$ weakly dominates $\noalert$ for node~$i$. 
And $\alert$ is strictly dominant if $\beta_i<\penaltySlash + \frac{\operatorCost}{n}$.
\end{claim}

\begin{proof}
Consider a node~$i$ and let~$y\in\{0,1,\ldots,n-1\}$ be the number of other alerters.
We consider two cases based on the value of~$y$.
\begin{enumerate}
  \item \textbf{$y=0$.} If nobody else alerts, then~$
u_i(\noalert,\beta_i,\alerters)=\beta_i$ and ~$u_i(\alert,\beta_i,\alerters)=\penaltySlash(n-1) + \operatorCost$.
Since $n\ge 2$ and $\beta_i\leq\penaltySlash + \frac{\operatorCost}{n}$, we have $\penaltySlash(n-1) + \operatorCost \geq \beta_i$, so $\alert$ is better, and is strictly better when $\beta_i<\penaltySlash + \frac{\operatorCost}{n}$.

\item \textbf{$y\ge 1$.} If at least one other node alerts, then~${u_i(\noalert,\beta_i,\alerters)=\beta_i-\penaltySlash< \frac{\operatorCost}{n}}$ and ${u_i(\alert,\beta_i,\alerters)=\frac{\penaltySlash(n-1-y) + \operatorCost}{y+1}\ge \frac{\operatorCost}{n}}$.
Thus $\alert$ is better.
\end{enumerate}

Since $\alert$ yields a higher utility in every case, it dominates $\noalert$ for all players, and is strictly dominant when $\beta_i<\penaltySlash + \frac{\operatorCost}{n}$.
\end{proof}

\begin{corollary}[Bribery is unprofitable when $\advGain < \penaltySlash \, n + \operatorCost$]
If the adversary's gain from a successful attack is less than the total penalty incurred by all nodes~$\advGain < \penaltySlash \, n + \operatorCost$, then bribery is unprofitable, and the adversary's optimal strategy is to not bribe any node.
\end{corollary}

\begin{proof}
Assume the adversary's gain from a successful attack is~$\advGain < \penaltySlash \, n + \operatorCost$.
If she offers bribes~${\beta_i > \penaltySlash + \frac{\operatorCost}{n}}$ to all nodes, then her expected utility is either~${\advGain - \sum_{i=1}^{n} \beta_i <0}$ or~${- \sum_{i \in \mathcal{N} \setminus \alerters} \beta_i <0}$.
In both cases, bribery is unprofitable.
Therefore, her total bribery is smaller than~$\penaltySlash \, n + \operatorCost$, and there exists at least one node~$i$ whose bribe satisfies~$\beta_i \leq \penaltySlash + \frac{\operatorCost}{n}$.
By Claim~\ref{clm:alert-dominant_opcost} node~$i$'s dominant strategy is to choose $\alert$.
In that case, the attack fails and the adversary's expected utility is~$- \sum_{i \in \mathcal{N} \setminus \alerters} \beta_i$.
Therefore, the optimal strategy for the adversary is to not bribe any nodes, i.e., to choose $\beta_i=0$ for all nodes~$i$.
We showed that in both cases, bribery is unprofitable and the optimal strategy for the adversary is to not bribe any node.
\end{proof}

We showed that if the offered bribe~$\beta_i\leq\penaltySlash$, or~$\beta_i\geq \penaltySlash(n-1)$, then there exists a dominant strategy for the node~$i$.
However, we now show that if the offered bribes satisfy~$\penaltySlash<\beta_i<\penaltySlash(n-1)$, then no symmetric pure Nash equilibrium exists.

\begin{lemma}[No symmetric pure equilibrium]
\label{lem:no-symm-pure_opcost}
 Consider the second phase of the simultaneous alerting game 
with~$n\ge2$ and a~$\penaltySlash>0$.
If there exists a node~$i$ whose bribe~$\beta_i$ satisfies

\begin{equation}
\label{eq:interior_opcost}
\penaltySlash  + \frac{\operatorCost}{n} \;<\; \beta_i \;<\; \penaltySlash(n-1) + \operatorCost,
\end{equation}
then there exists no symmetric pure Nash equilibrium.
\end{lemma}

\begin{proof}
There are only two symmetric pure profiles:
\begin{enumerate}
  \item \emph{All \alert: $\alerters=\mathcal N$.}
Then each node's payoff is~$\frac{\operatorCost}{n}$ (Equation~\ref{eq:sim_payoff_opcost}).
A unilateral deviation by node~$i$ to \noalert yields a payoff of~$\beta_i-\penaltySlash$ (Equation~\ref{eq:sim_payoff_opcost}) since~$\alerters\setminus\{i\}\neq\emptyset$.
By the left inequality in Equation~\ref{eq:interior_opcost},~$\beta_i-\penaltySlash>\frac{\operatorCost}{n}$, so deviating is strictly profitable.
Hence all~\alert is not a Nash equilibrium.

  \item  \emph{All \noalert: $\alerters=\emptyset$.}
Node~$i$'s payoff is~$\beta_i$.
A unilateral deviation by node~$i$ to~$\alert$ yields a payoff of~$\penaltySlash\frac{|\,\mathcal N\setminus\{i\}\,|}{|\,\{i\}\,|} + \operatorCost=\penaltySlash(n-1) + \operatorCost$.
By the right inequality in Equation~\ref{eq:interior_opcost},~$\penaltySlash(n-1) + \operatorCost>\beta_i$, so deviating is strictly profitable.
Hence all~\noalert is not a Nash equilibrium.
\end{enumerate}

These are the only symmetric pure profiles and neither is a Nash equilibrium.
\end{proof}

\subsection{Mixed Strategy Nash Equilibria}
\label{sec:mixed-strategy-equilibria_opcost}
By Lemma~\ref{lem:no-symm-pure_opcost}, whenever a bribe satisfies~${\penaltySlash + \frac{\operatorCost}{n} < \beta_i < \penaltySlash(n-1) + \operatorCost}$, no symmetric pure Nash equilibrium exists.
Since a node that is offered such a bribe has no dominant strategy, it may choose its action with some probability. 
We study the existence of mixed-strategy equilibria in this interior range of bribes, where the strategy of a node is its probability of choosing $\alert$.
In the next lemma, we give a lower bound on the adversary's \emph{expected} total bribe outlay that holds for any bribe vector and any strategy profile of the nodes.

\begin{lemma}[Lower bound on expected bribe]
\label{lem:bribe-lower-bound_opcost}
For any bribe vector~$(\beta_1,\ldots,\beta_n)$,
let each node $i$ alert with probability $p_i\in[0,1]$ and not alert with probability $q_i=1-p_i$, and assume those probabilities are independent.
Denote by  
\[
Q \;:=\; \Pr(\alerters = \emptyset) \;=\; \prod_{i=1}^n q_i 
\]
 the probability that no node alerts.

Then, in any mixed strategy equilibria, the expected total bribe paid by the adversary satisfies
\begin{equation}
\label{eq:lower-bound-bribes_opcost}
\sum_{i=1}^n \beta_i\, q_i
\;\ge\;
(\penaltySlash\, n(n-1) + n \,\operatorCost )\, Q.
\end{equation}
\end{lemma}

\begin{proof}
Consider a node~$i$ and let~$Y_{-i}$ be the number of other alerters (whether or not~$i$ alerts). 
Denote by~${Q_{-i}:=\Pr(Y_{-i}=0)=\prod_{j\ne i} q_j}$ the probability that other nodes do not alert.
If~$i$ alerts while~$Y_{-i}=y$, its payoff is~${\frac{\penaltySlash(n-(y+1))+ \operatorCost}{y+1}=\penaltySlash(\frac{n}{1+y}-1) + \frac{\operatorCost}{1+y}}$,
so
\begin{equation}
\label{eq:alert-exp_opcost}
\begin{aligned}
\mathbb{E}\!\big[\rho_i(a_i = \alert,\beta_i,\alerters)\big]
&= \penaltySlash\!\left(
    n\,\mathbb{E}\!\Big[\frac{1}{1+Y_{-i}}\Big] - 1
\right) \\
&\quad + \operatorCost\,\mathbb{E}\!\Big[\frac{1}{1+Y_{-i}}\Big].
\end{aligned}
\nonumber
\end{equation}

We denote this expected payoff by $U_i^{\alert}$.
If $i$ chooses $\noalert$, it always receives $\beta_i$ and pays the penalty if and only if someone else alerts:
\begin{equation}
\label{eq:noalert-exp_opcost}
\begin{aligned}
\mathbb{E}[\rho_i(a_i = \noalert,\beta_i,\alerters)]
&= \beta_i - \penaltySlash\,\Pr(Y_{-i}\ge1) \\
&= \beta_i + \penaltySlash\,Q_{-i} - \penaltySlash.
\end{aligned}
\nonumber
\end{equation}
We denote this expected payoff by $U_i^{\noalert}(\beta_i)$.

Notice that~$U_i^{\alert}$ is independent of~$\beta_i$ while~$U_i^{\noalert}(\beta_i)$ is strictly increasing in~$\beta_i$.
In equilibrium, for each node~$i$, its utility from alerting and not alerting must be equal. 
I.e., $U_i^{\alert} = U_i^{\noalert}(\beta_i)$.
Otherwise, a node would deviate to the action with the higher expected utility.
Thus, in equilibrium, it holds that 
\[
\penaltySlash(n\,\mathbb{E}\!\Big[\frac{1}{1+Y_{-i}}\Big]-1) + \operatorCost \mathbb{E}\!\Big[\frac{1}{1+Y_{-i}}\Big]
= \beta_i + \penaltySlash\,Q_{-i} - \penaltySlash.
\]

Equivalently, 
\begin{equation}
\beta_i = \penaltySlash\!\left(n\,\mathbb{E}\!\Big[\frac{1}{1+Y_{-i}}\Big]-Q_{-i}\right) + \operatorCost \mathbb{E}\!\Big[\frac{1}{1+Y_{-i}}\Big].
\nonumber
\end{equation}

Since $\tfrac{1}{1+y}\ge 1$ when $y=0$ and $\tfrac{1}{1+y}\ge0$ otherwise, we have
\begin{equation}
  \begin{aligned}
\mathbb{E}\!\left[\frac{1}{1+Y_{-i}}\right]
&=\sum_{y=0}^{n-1}\frac{1}{1+y}\,\Pr(Y_{-i}=y)\\
&\;\ge\; 1\cdot\Pr(Y_{-i}=0)+0\cdot\Pr(Y_{-i}\ge 1)\\
&=\Pr(Y_{-i}=0)=Q_{-i},
\end{aligned}
\nonumber
\end{equation}

and thus
\begin{equation}
\label{eq:beta-simple-lb_opcost}
\beta_i \;\ge\; (\penaltySlash\,(n-1)+ \operatorCost)\,Q_{-i}
\end{equation}
Multiplying Equation~\ref{eq:beta-simple-lb_opcost} by~$q_i$ and using~${q_i Q_{-i}=\prod_{j=1}^n q_j=Q}$ yields~$\beta_i q_i \ge (\penaltySlash\,(n-1)+ \operatorCost)\,Q$ for all such $i$.
Summing over $i=1,\ldots,n$, we get
\[
\sum_{i=1}^n \beta_i q_i \;\ge\; (\penaltySlash\,(n-1)+ \operatorCost)\,\sum_{i=1}^n Q
\;=\; (\penaltySlash\,n(n-1)+ \operatorCost n)\,Q,
\]
concluding the proof.
\end{proof}

In the boundary cases where a dominant strategy exists, if all bribes are at least~$\penaltySlash(n-1)$, then all nodes choose $\noalert$, i.e., $q_i=1$ for all $i$ and $Q=1$.
Inequality~\ref{eq:lower-bound-bribes_opcost} holds.
If all bribes are at most~$\penaltySlash$, then all nodes choose $\alert$, i.e., $q_i=0$ for all $i$ and $Q=0$.
Inequality~\ref{eq:lower-bound-bribes_opcost} holds as well.

To analyze the adversary's optimal strategy, we consider her expected payoff.
The adversary's expected payoff is her gain from a successful attack times the probability that no node alerts, minus the expected total bribe paid to the nodes.
I.e.,~${\mathbb{E}[\rho_{\text{adv}}]=\advGain\,Q-\sum_i \beta_i q_i}$.
By Lemma~\ref{lem:bribe-lower-bound_opcost}, we get the upper bound~${\mathbb{E}[\rho_{\text{adv}}]\le (\advGain-\penaltySlash\,n(n-1))\,Q}$,
which is nonpositive for~$\advGain\le \penaltySlash\,n(n-1)$.

Therefore, if the adversary's gain from a successful attack is less than~$\penaltySlash\,n(n-1)$, then her expected payoff is nonpositive for any bribe vector at any mixed-strategy equilibrium of the nodes.
So her optimal strategy is to not bribe any node, yielding zero expected payoff.

\section{Simultaneous Alerting Game - \alert Dominance for Low Bribes}
\label{app:alert-dominance}

\begin{restate}{Claim~\ref{clm:alert-dominant}}
\lemmaAlertDominance
\end{restate}

\begin{proof}
Consider a node~$i$ and let~$y\in\{0,\ldots,n-1\}$ be the number of other alerters.
We consider two cases based on the value of~$y$.
\begin{enumerate}
  \item \textbf{$y=0$.} If nobody else alerts, then~$
u_i(\noalert,\beta_i,\alerters)=\beta_i$ and ~$u_i(\alert,\beta_i,\alerters)=\penaltySlash(n-1)$.
Since $n\ge 2$ and $\beta_i\leq\penaltySlash$, we have $\penaltySlash(n-1)\geq \beta_i$, so $\alert$ is better, and is strictly better when $\beta_i<\penaltySlash$.

\item \textbf{$y\ge 1$.} If at least one other node alerts, then~${u_i(\noalert,\beta_i,\alerters)=\beta_i-\penaltySlash<0}$ and ${u_i(\alert,\beta_i,\alerters)=\penaltySlash\frac{n-1-y}{y+1}\ge 0}$.
Thus $\alert$ is better.
\end{enumerate}

Since $\alert$ yields a higher utility in every case, it dominates $\noalert$ for all players, and is strictly dominant when $\beta_i<\penaltySlash$.
\end{proof}

\begin{restate}{Claim~\ref{cor:unprofitable-bribery}}
\corUnprofitableBribery
\end{restate}

\begin{proof}
Assume the adversary's gain from a successful attack is~$\advGain < \penaltySlash \, n$.
If she offers bribes~${\beta_i > \penaltySlash}$ to all nodes, then her expected utility is either~${\advGain - \sum_{i=1}^{n} \beta_i <0}$ or~${- \sum_{i \in \mathcal{N} \setminus \alerters} \beta_i <0}$.
In both cases, bribery is unprofitable.
Therefore, her total bribery is smaller than~$\penaltySlash \, n$, and there exists at least one node~$i$ whose bribe satisfies~$\beta_i \leq \penaltySlash$.
By Claim~\ref{clm:alert-dominant} node~$i$'s dominant strategy is to choose $\alert$.
In that case, the attack fails and the adversary's expected utility is~$- \sum_{i \in \mathcal{N} \setminus \alerters} \beta_i$.
Therefore, the optimal strategy for the adversary is to not bribe any nodes, i.e., to choose for all nodes~$i$, $\beta_i=0$.
We showed that in both cases, bribery is unprofitable and the optimal strategy for the adversary is to not bribe any node.
\end{proof}

\section{Simultanious Alerting Game - Mixed Strategy Equilibria}
\label{app:mixed-strategy-equilibria}

\begin{restate}{Lemma~\ref{lem:bribe-lower-bound}}
\lemmaLowerBoundBribery
\begin{equation}
\label{eq:lower-bound-bribes2}
\sum_{i=1}^n \beta_i\, q_i
\;\ge\;
\penaltySlash\, n(n-1)\, Q.
\end{equation}
\end{restate}

\begin{proof}
Consider a node~$i$ and let~$Y_{-i}$ be the number of other alerters (whether or not~$i$ alerts). 
Denote by~${Q_{-i}:=\Pr(Y_{-i}=0)=\prod_{j\ne i} q_j}$ the probability that other nodes do not alert.
If~$i$ alerts while~$Y_{-i}=y$, its payoff is~${\penaltySlash\,\frac{n-(y+1)}{y+1}=\penaltySlash(\frac{n}{1+y}-1)}$,
so
\begin{equation}
\label{eq:alert-exp}
\mathbb{E}[\rho_i(a_i = \alert,\beta_i,\alerters)]
=\penaltySlash(n\,\mathbb{E}\!\Big[\frac{1}{1+Y_{-i}}\Big]-1).
\nonumber
\end{equation}
We denote this expected payoff by $U_i^{\alert}$.
If $i$ chooses $\noalert$, it always receives $\beta_i$ and pays the penalty if and only if someone else alerts:
\begin{equation}
\label{eq:noalert-exp}
\begin{aligned}
\mathbb{E}[\rho_i(a_i = \noalert,\beta_i,\alerters)]
&= \beta_i - \penaltySlash\,\Pr(Y_{-i}\ge1) \\
&= \beta_i + \penaltySlash\,Q_{-i} - \penaltySlash.
\end{aligned}
\nonumber
\end{equation}
We denote this expected payoff by $U_i^{\noalert}(\beta_i)$.

Notice that~$U_i^{\alert}$ is independent of~$\beta_i$ while~$U_i^{\noalert}(\beta_i)$ is strictly increasing in~$\beta_i$.
In equilibrium, for each node~$i$, its utility from alerting and not alerting must be equal. 
I.e., $U_i^{\alert} = U_i^{\noalert}(\beta_i)$.
Otherwise, a node would deviate to the action with the higher expected utility.
Thus, in equilibrium, it holds that 
\[
\penaltySlash(n\,\mathbb{E}\!\Big[\frac{1}{1+Y_{-i}}\Big]-1) = \beta_i + \penaltySlash\,Q_{-i} - \penaltySlash.
\]

Equivalently, 
\begin{equation}
\beta_i = \penaltySlash\!\left(n\,\mathbb{E}\!\Big[\frac{1}{1+Y_{-i}}\Big]-Q_{-i}\right).
\nonumber
\end{equation}

Since $\tfrac{1}{1+y}\ge 1$ when $y=0$ and $\tfrac{1}{1+y}\ge0$ otherwise, we have
\begin{equation}
  \begin{aligned}
\mathbb{E}\!\left[\frac{1}{1+Y_{-i}}\right]
&=\sum_{y=0}^{n-1}\frac{1}{1+y}\,\Pr(Y_{-i}=y)\\
&\;\ge\; 1\cdot\Pr(Y_{-i}=0)+0\cdot\Pr(Y_{-i}\ge 1)\\
&=\Pr(Y_{-i}=0)=Q_{-i},
\end{aligned}
\nonumber
\end{equation}

and thus
\begin{equation}
\label{eq:beta-simple-lb}
\beta_i \;\ge\; \penaltySlash\,(n-1)\,Q_{-i}
\end{equation}
Multiplying Equation~\ref{eq:beta-simple-lb} by~$q_i$ and using~${q_i Q_{-i}=\prod_{j=1}^n q_j=Q}$ yields~$\beta_i q_i \ge \penaltySlash\,(n-1)\,Q$ for all such $i$.
Summing over $i=1,\ldots,n$, we get
\[
\sum_{i=1}^n \beta_i q_i \;\ge\; \penaltySlash\,(n-1)\,\sum_{i=1}^n Q
\;=\; \penaltySlash\,n(n-1)\,Q,
\]
concluding the proof.
\end{proof}

\begin{algorithm}[t]
\caption{$\penaltySlash$-Trusted-Hardware-Based Alerting: Contract Side ($\AlertsCon$)}
\label{alg:contract}
\KwIn{$\CommitBlocks=\lceil \CommitInterval/\blockInterval\rceil$, $\RevealBlocks=\lceil \RevealInterval/\blockInterval\rceil$ with $\CommitBlocks<\RevealBlocks$; penalty $\penaltySlash$.}
\SetAlgoNoLine
\DontPrintSemicolon

\nonl \textbf{Round start:}\;
Record current height $H_0$; set reveal barrier $H^\star \leftarrow H_0+\CommitBlocks$\;
Define commit window $[H_0,H^\star)$ and reveal window $[H^\star,H^\star+\RevealBlocks)$\;
Initialize arrays $\textit{commit}[], \textit{reveal}[]$ to $\bot$\;

\nonl \textbf{On \emph{commit tx} from node $i$ carrying $\committment_i$:}\;
\eIf{$\text{height} \notin [H_0,H^\star)$ \textbf{or} node $i$ already committed}{
    reject\;
}
{
Store $\textit{commit}[i]\leftarrow (\committment_i,\text{height})$\;
}
\nonl \textbf{On \emph{reveal tx} from node $i$ carrying $(a,\Pi_\alert,r,A)$:}\;
$\committment_i, \text{height} \leftarrow \textit{commit}[i]$\;
\If{$\text{height} \notin [H^\star,H^\star+\RevealBlocks)$ \textbf{or} $\textit{commit}[i]=\bot$ \textbf{or} attestation $A$ invalid or not binding $(\committment_i,\CommitBlocks)$ \textbf{or} $\committment_i \ne \text{Hash}(m,r,\CommitBlocks)$ with $m=(a \mathbin\Vert \text{Hash}(\Pi_\alert))$}{
    reject\;
}

\If{$a=\alert$}{
    \If{$\text{Hash}(\Pi_\alert)\ne h_\alert$ \textbf{or} $\Pi_\alert$ invalid}{
        reject\;
    }
}
Store $\textit{reveal}[i]\leftarrow (a,\Pi_\alert,r,A)$\;

\nonl \textbf{Settlement at end of reveal window:}\;
Let $\alerters \leftarrow \{i : \textit{reveal}[i]\neq \bot \text{ and has } a=\alert\}$\;
Let $N_0 \leftarrow \{i : i\notin \alerters\}$ \tcp*{includes no-commit or no-reveal}
Let $N_1 \leftarrow \{i : \textit{reveal}[i]\neq \bot \text{ and has } a=\noalert\}$\;
\eIf{$\alerters\neq \emptyset$}{
    \ForEach{$j \in N_0$}{ slash $j$ by $\penaltySlash$ }
    \ForEach{$i \in \alerters$}{ pay $i$ the amount $\frac{\penaltySlash \cdot |N_0|}{|\alerters|}$ }
          Invoke \textbf{Alert()}
}{
\If{$|N_0 \setminus N_1| > 0$ \tcp*[h]{some nodes committed but did not reveal}}{
    Invoke \textbf{Alert()}
}

}
\end{algorithm}

\begin{algorithm}[t]
\caption{Per-Node~$i$ Procedure}
\label{alg:node}
\KwIn{Reveal barrier $H^\star$ from \AlertsCon; $\CommitBlocks$; reveal window $[H^\star,H^\star+\RevealBlocks)$; chain headers; PoP parameters.}
\BlankLine
\DontPrintSemicolon
\SetAlgoNoLine
\nonl \textbf{Commit phase ($[H_0,H^\star)$):}\;
Choose action $a_i \in \{\alert,\noalert\}$\;
\eIf{$a_i=\alert$}{
    form alert proof $\Pi_\alert$ and $h_\alert\leftarrow \text{Hash}(\Pi_\alert)$\;
}{
    set $\Pi_\alert \leftarrow \bot$ and $h_\alert\leftarrow \text{Hash}(\bot)$\;
}

$\committment_i \leftarrow \tee{\text{TEE\_}\Seal{m = (a_i \mathbin\Vert h_\alert),\CommitBlocks}}$ \tcp*{outputs commitment~$\committment_i$}
Send \emph{commit tx} carrying $\committment_i$ to \AlertsCon\;

\nonl \textbf{PoP anchoring:}\;
Node invokes TEE to update internal checkpoint to $C_B$.

\nonl \textbf{Reveal phase ($[H^\star,H^\star+\RevealBlocks)$):}\;
Obtain a PoP $\Pi_{\textit{PoP}}=(C_B,B_1,\dots,B_n)$ s.t. $\committment_i$ is included in some $B_j$ and $n-j \ge \CommitBlocks$ \;
$(m,r),A \leftarrow \tee{\text{TEE\_}\UnsealAfter{\committment_i,\CommitBlocks,\Pi_{\textit{PoP}}}}$\;
\If{TEE returns failure}{ abort (penalized as non-revealer)\;}
Send \emph{reveal tx} carrying $(a_i,\Pi_\alert,r,A)$\;
\end{algorithm}

\section{TEE-based Protocol - Algorithms}
\label{app:tee_based_protocol}
In Algorithm~\ref{alg:contract}, we present the smart contract~\AlertsCon that implements the TEE-based alerting protocol described in Section~\ref{sec:hardware_based_protocol}.
And in Algorithm~\ref{alg:node}, we present the node-side algorithm that implements the TEE-based alerting protocol described in Section~\ref{sec:hardware_based_protocol}.

\section{TEE-based protocol - Alert Capability}
\label{app:tee_alert_capability}
\begin{restate}{Lemma~\ref{lem:timing-feasibility}}
  \lemmaTEETimingFeasibility
\end{restate}

\begin{proof}
  We show that all nodes~$i$ can have a commit transaction included in the chain during the commit window~$[H_0, H^\star)$ and
   can later obtain a valid PoP for its commitment and have a reveal transaction included during the reveal window~$[H^\star, H^\star + \RevealBlocks)$.

We prove separately for the commit and reveal phases (see Figure~\ref{fig:commit_reveal}).
The commit window lasts for $\CommitBlocks$ blocks: it starts at height $H_0$ and ends just before
height $H^\star = H_0 + \CommitBlocks$.
In time, the commit window is~${[t_0, t_0 + \CommitInterval)}$, where
$t_0$ is the time corresponding to height $H_0$.
Recall, $\CommitInterval \ge \NetworkWrite$.

Consider any node $i$.
Local computation to form a commitment~$\committment_i$ is negligible compared to~$\NetworkWrite$.
If node $i$ sends her transaction at any time~$ {t \in [t_0,\, t_0 + \CommitInterval - \NetworkWrite]}$,
then by the synchrony assumption the transaction is included in some block by time~${t + \NetworkWrite \le t_0 + \CommitInterval}$, i.e.,
before the end of the commit window.

For the reveal phase, consider node $i$ that committed successfully at some height $h_c$ with~$H_0 \le h_c < H^\star = H_0 + \CommitBlocks.$
The TEE requires a PoP showing that at least $\CommitBlocks$ blocks have been added after $\committment_i$.
Therefore, node $i$ must wait until the chain reaches a height of~$h_c + \CommitBlocks$.

The reveal window is between blocks
\[
  [H^\star, H^\star + \RevealBlocks)
  = [H_0 + \CommitBlocks, H_0 + \CommitBlocks + \RevealBlocks).
\]
From $h_c < H_0 + \CommitBlocks$ we have
\[
  h_c + \CommitBlocks \;<\; H_0 + 2\CommitBlocks.
\]
Since $\CommitBlocks < \RevealBlocks$, we get that
\[
  H_0 + \CommitBlocks + \RevealBlocks
  > H_0 + 2\CommitBlocks.
\]
Hence~$h_c + \CommitBlocks
  < H_0 + \CommitBlocks + \RevealBlocks $,
so the height $h_c + \CommitBlocks$ is always reached strictly before
the end of the reveal window.

Once the chain has reached height $h_c + \CommitBlocks$, 
 PoP guarantees that node $i$ can get a valid PoP
$\Pi_{\mathit{PoP}}$ in negligible time and call
\[
  (m,r),A \gets \mathsf{Unseal\mbox{-}After}(\committment_i,\CommitBlocks,\Pi_{\mathit{PoP}}).
\]
The node then sends a reveal transaction containing~$(a,\Pi_\alert,h_\alert,r,A)$.

In terms of time, the reveal window has duration~$\RevealInterval$, and
by assumption $\RevealInterval \ge \NetworkWrite$.
Therefore, any reveal transaction that a node sent at least~$\NetworkWrite$
steps before the end of the reveal window is included on-chain before the end of the reveal phase, by the same synchrony argument as in the commitment phase.

Therefore, a node can always raise an alert before the end of the alerting period, so the protocol achieves Alert Capability.
\end{proof}

\section{TEE-based protocol achieves Reward Distribution}
\label{app:tee_reward_dist_proof}
\begin{restate}{Lemma~\ref{lem:payoff-equivalence}}
    \lemmaTEEPayoffEquivalence
\end{restate}

\begin{proof}
Consider an execution of the TEE-based alerting protocol and let~$\alerters$ be the set of nodes that revealed action~$\alert$ with valid proofs.
We analyze the payoff of a node~$i$ based on its induced action~$a_i$.
If~$i \in \alerters$, then by the protocol, node~$i$ receives a reward of~$\penaltySlash \frac{ (n - |\alerters|)}{|\alerters|}$, matching the payoff in Equation~\ref{eq:sim_payoff}.
If~$i \notin \alerters$, then node~$i$ either revealed action~$\noalert$, or deviated from the protocol.
In both cases, by the protocol, node~$i$ is penalized with~$\penaltySlash$ if~$\alerters \neq \emptyset$, matching the payoff in Equation~\ref{eq:sim_payoff}.
If~$\alerters = \emptyset$, then no nodes are penalized or rewarded, matching the payoff in Equation~\ref{eq:sim_payoff}.
For the adversary, if~$\alerters = \emptyset$, then by the protocol, then no node raised an alert (either by revealing~$\noalert$ or deviating), so the adversary gains~$\advGain - \sum_{i \in \mathcal{N}} \beta_i$, matching the payoff in Equation~\ref{eq:sim_payoff_adv}.
Otherwise,~$\alerters \neq \emptyset$, then at least one node raised an alert, so the adversary does not gain anything, but still pays the bribes, resulting in a payoff of~$- \sum_{i \in \mathcal{N}} \beta_i$, matching the payoff in Equation~\ref{eq:sim_payoff_adv}.
\end{proof}

\section{Deterministic Node Sequencing}
\label{app:node_sequencing}
The protocol uses a straightforward lexicographic next-permutation algorithm~\cite{Worlton1968TheAO} to sequence the nodes.
Each node and the smart contract maintain the permutation state representing the current order of nodes. 
The sequence begins with the canonical permutation~$\pi_1 = (1,2,\ldots,n)$, corresponding to the lexicographically smallest ordering of all~$n$ nodes. 
Both nodes and the contract implement the same algorithm: given the current permutation~$\pi_t$, the function~$\nextperm{\cdot}$ computes~${\pi_{t+1}= \nextperm{\pi_t}}$ by identifying the rightmost pair~$(i,j)$ such that~$\pi_t(i) < \pi_t(j)$, swapping these elements, and reversing the suffix following index $i$. 
This operation deterministically produces the next permutation in lexicographic order and requires no communication among nodes or with the smart contract.
This ensures that all nodes and the contract have a consistent view of the current permutation.

Before each alerting period~$t$, all nodes and the contract independently advance to the next permutation using the function~$\pi_{t} = \nextperm{\pi_{t-1}}$ and nodes alert based on this ordering.
When a node alerts, the contract verifies that the alerting node's identifier corresponds to the node assigned to the current slot in the permutation.
To do so, denote by~$t_0$ the starting time of the alerting period, and by~$t_\alert$ the time when the alert is raised.
the current slot is given by $s = \left\lceil \frac{t_\alert-t_0}{\SlotLen} \right\rceil$, so the contract verifies that the alerting node is~$\pi_t(s)$.

\section{Random Node Sequencing}
\label{sec:random_node_sequencing}
In this protocol, we focus on randomly sequencing the nodes as it offers improved protection against denial of service attacks.
The transmission sequence for the nodes is private to the nodes only, but our protocol requires a way to present the sequence~$\pi$ for a given alert verifiably on-chain. 

The sequence is generated from a private \textit{sequencing seed}~$\sk$ input to a VRF, and a function~$\seq$ that uses the VRF's output. 
The VRF public key~$\pk$ is stored in the smart contract \AlertsCon, and each node holds the secret key~$\sk$.
The public key~$\pk$ is stored in \AlertsCon, and each node holds the secret key~$\sk$.
Here,~$\seq_{\sk}()$ specifies the sequencing of the~$n$ nodes for a given alerting period.
That is,~${\seq_{\sk}() \in S_n}$ is a pseudorandom permutation on the~$n$ nodes indices, for example generated using a Fisher-Yates shuffle \cite{knuth1973taocp}.

We denote by~${\sf proof} \leftarrow \seq_{\sk}().{\sf proof}$ the proof of correctness generated for output~$\pi_{} \leftarrow \seq_{\sk}()$. 
Furthermore, the algorithm~$\seq_{\sk}().{\sf verify}(\pi,{\sf proof})$ returns either~$\{{\tt true}$ or~${\tt false}\}$ verifying the correctness of ${\sf proof}$ for a claimed output $\pi$ of $\seq_{\sk}()$. 
The sequencing seed~$\sk$ is shared among all nodes, each node independently computes the VRF output to determine its slot in the next publishing period.

\subsubsection{Raising an Alert}
\label{sec:random_round_begin}
During its assigned slot, the node can raise an alert if it suspects one is neededt.
To do so, the node sends an alert transaction to \AlertsCon. 
The alert transaction includes the permutation~$\pi$ for the current alerting period, and a proof that the node is assigned to the current slot~$s$ in the permutation~$\pi$.
We denote the alert transaction by~$\txalert = \{\pi, {\sf proof} \}$.

Before checking the alert validity, the smart contract $\AlertsCon$ verifies that the alerting node is assigned to the current slot~$s$ in the permutation~$\pi$ using the proof~${\sf proof}$ and the function~$\seq_{\sk}().{\sf verify}(\pi,{\sf proof})$. 
If so the alerting period ends, and the contract proceeds to slashing and rewarding as follows.

Let~$i$ be the node that raised the alert, and let~$s$ be the slot assigned to node~$i$ in the permutation~$\pi$ i.e., $\pi(s) = i$.
The contract~\AlertsCon slashes all nodes~$j$ that had the opportunity to alert but did not do so, i.e., all nodes~$j$ that~$s_j < s$.
And rewards the alerting node with~${\sum_{j=1}^{s-1}\penaltySlash = \penaltySlash (s-1)}$ for raising the alert.

\section{Sequential Protocol - Adversary's Optimal Strategy}
\label{app:sequential_adversary_optimal}

\begin{restate}{Corollary~\ref{cor:profitable-bribery}}
\corProfitableBribery
\end{restate}

\begin{proof}
  The proof follows from Proposition~\ref{claim:budget_threshold} and the fact that the adversary will not bribe any node not to alert if she cannot bribe all nodes to not alert.
  The adversary is rational and utility maximizing.
  By proposition~\ref{claim:budget_threshold}, if~$\advGain < \penaltySlash \frac{n(n-1)}{2}$, then at least one node will alert.
  In this case, in an execution where~$k$ is the first node to alert, the adversary's utility is~$-\sum_{i=1}^{k-1} \beta_i \leq 0$.
  Thus, to maximize her utility, the adversary's optimal strategy is to choose all bribes to be~$\beta_i = 0$ for all nodes~$i \in \mathcal{N}$.
\end{proof}

\section{Sequential Protocol - Adversary's Bribery Cost}
\label{app:sequential_bribery_cost}
\begin{restate}{Corollary~\ref{cor:sequential_bribery_cost}}
\corBriberyCost
\end{restate}

\begin{proof}
  By applying Lemma~\ref{claim:adversary_gain_threshold_helper} to all nodes, the adversary can bribe all nodes in slots $1, \dots, n$ to not alert if and only if her gain from a missing alert~$\advGain$ is larger than~$\sum_{i=1}^{n}\penaltySlash(i - 1) = \penaltySlash \frac{n(n - 1)}{2}$.
  Since bribing all nodes to not alert requires over~$\penaltySlash \frac{n(n-1)}{2}$, the attack cost grows quadratically with the number of nodes~$n$.
\end{proof}

\end{document}

%% file: my_macros.tex
\usepackage{lipsum}    
\newif\iflong

\SetKwIF{If}{ElseIf}{Else}{if}{}{else if}{else}{}%

\SetKwFor{ForEach}{for each}{}{}%
\SetKwFor{For}{for}{}{}%
\SetKwFor{While}{while}{}{}%

\newcommand{\seq}{{\sf seq}\xspace}
\newcommand{\sk}{\ensuremath{{\sf sk}}\xspace}
\newcommand{\pk}{\ensuremath{{\sf pk}}\xspace}

\newcommand{\tx}{\ensuremath{\textit{tx}}\xspace}
\newcommand{\AlertsCon}{\ensuremath{{\textbf{AlertingContract}}}\xspace}

\newcommand{\txalert}{\ensuremath{{\sf tx_{alert}}}\xspace}

\newcommand{\nextperm}[1]{\ensuremath{\textit{next\_perm}({#1})}\xspace}
\newcommand{\checkreport}[1]{\ensuremath{\textit{validate\_alert}({#1})}\xspace}

\newcommand{\penaltySlash}{{{\lambda}}\xspace}

\newcommand{\alertReward}[1]{{r({#1})}\xspace}

\newcommand{\CommitInterval}{\ensuremath{\Delta_{\sf commit}}\xspace}
\newcommand{\CommitBlocks}{\ensuremath{N_{\sf commit}}\xspace}
\newcommand{\RevealBlocks}{\ensuremath{N_{\sf reveal}}\xspace}
\newcommand{\RevealInterval}{\ensuremath{\Delta_{\sf reveal}}\xspace}

\newcommand{\NetworkWrite}{\ensuremath{\Delta_{\sf write}}\xspace}
\newcommand{\blockInterval}{\ensuremath{\Delta_{\sf block}}\xspace}
\newcommand{\advGain}{\ensuremath{G}\xspace}
\newcommand{\SlotLen}{\ensuremath{\Delta_{\sf slot}}\xspace}

\makeatletter
\renewcommand{\Function}[2]{%
  \csname ALG@cmd@\ALG@L @Function\endcsname{#1}{#2}%
  \def\jayden@currentfunction{#1}%
}
\newcommand{\funclabel}[1]{%
  \@bsphack
  \protected@write\@auxout{}{%
    \string\newlabel{#1}{{\jayden@currentfunction}{\thepage}}%
  }%
  \@esphack
}
\makeatother

\newcommand{\tr}{\ensuremath{\textbf{True}}\xspace}
\newcommand{\fa}{\ensuremath{\textbf{False}}\xspace}
\newcommand{\data}{\ensuremath{d}\xspace}

\newcommand{\alert}{\ensuremath{\mathsf{Alert}}\xspace}
\newcommand{\noalert}{\ensuremath{\mathsf{NoAlert}}\xspace}
\newcommand{\alerters}{\ensuremath{{F}}\xspace}
\newcommand{\operatorCost}{\ensuremath{{c}}\xspace}

\newcommand{\committment}{\ensuremath{{\gamma}}\xspace}

\newcommand{\Seal}[1]{\ensuremath{\textit{Seal}({#1})}\xspace}
\newcommand{\UnsealAfter}[1]{\ensuremath{\textit{UnsealAfter}({#1})}\xspace}

\newcommand{\tee}[1]{{\color{blue}{#1}}}

\newenvironment{restate}[1]{\begin{trivlist} \item {\bf #1 (restated)}. 
  \em} {\end{trivlist}}

\newcommand{\lemmaTEEPayoffEquivalence}{
  The TEE-based protocol achieves Reward Distribution.
  }

\newcommand{\lemmaAlertDominance}{
In the simultaneous alerting game, if the adversary's offered bribe for a node~$i$ is~$\beta_i \leq \penaltySlash$, then the action~$\alert$ weakly dominates $\noalert$ for node~$i$. 
And $\alert$ is strictly dominant if $\beta_i<\penaltySlash$.
}

\newcommand{\corUnprofitableBribery}{
If the adversary's gain from a successful attack is less than the total penalty incurred by all nodes~$\advGain < \penaltySlash \, n$, then bribery is unprofitable, and the adversary's optimal strategy is to not bribe any node.
}

\newcommand{\lemmaLowerBoundBribery}{
For any bribe vector~$(\beta_1,\ldots,\beta_n)$,
let each node $i$ alert with probability $p_i\in[0,1]$ and not alert with probability $q_i=1-p_i$, and assume those probabilities are independent.
Denote the probability that no node alerts by  
\[
Q \;:=\; \Pr(\alerters = \emptyset) \;=\; \prod_{i=1}^n q_i .
\]

Then, in any mixed strategy equilibria, the expected total bribe paid by the adversary satisfies
}

\newcommand{\lemmaTEETimingFeasibility}{
The TEE-based protocol achieves Alert Capability.
  }

\newcommand{\corProfitableBribery}{
  If the adversary's gain from a missing alert~$\advGain$ is smaller than~$\penaltySlash\frac{n(n-1)}{2}$, then the adversary's optimal strategy is to not bribe any node.
}

\newcommand{\corBriberyCost}{
  The adversary bribes all nodes to not alert if and only if her gain from a missing alert~$\advGain$ is larger than~$\penaltySlash \frac{n(n-1)}{2}$.
  And the required bribe grows quadratically with the number of nodes.
}

%% file: in_progress.tex
\section{Introduction}

\label{sec:intro}


Blockchains are replicated state machines that maintain a public, append-only ledger.
\emph{Smart contracts} are stateful programs deployed on-chain that collectively secure over $\$170\text{B}$ in total value locked~\cite{coindesk-defi-tvl-2025}.
These contracts increasingly make high-stakes decisions that rely on timely \emph{alerts} about external events, including changes in financial asset prices~\cite{chainlink-feeds}, outcomes of votes~\cite{snapshot-safesnap}, and results of off-chain computations~\cite{eth-optimistic-rollups, arbitrum-sequencer}.
Those alerts are often required for triggering on-chain actions.
Delays, censorship, or errors in these alerts have already produced \emph{multi-billion-dollar} losses~\cite{chainalysis-hacking-2023,reuters-crypto-hacks-2024}.
Yet, in existing designs, suppressing alerts is cheap. 
An adversary can often bribe participants into silence at a cost that grows only linearly in the number of alerters, making denial-of-service attacks economically feasible at scale.

Most previous work~(\S\ref{sec:related_work}) examined mechanisms for incentivising alerts, but did not model the economics of suppressing them.
Existing systems such as optimistic rollups~\cite{eth-optimistic-rollups,arbitrum-sequencer,optimism-batcher}, oracle networks~\cite{chainlink-feeds}, and optimistic-oracle designs~\cite{uma-oo-overview,tellor-disputes} rely on participants to raise alerts, but do not analyze the cost of bribing those participants into silence.
Chen et al.~\cite{Prrr} study incentivizing timely on-chain alerts but also do not consider a bribing adversary.
Meanwhile, other works show that rational adversaries can use bribes to induce protocol deviations in blockchain systems~\cite{bribes1, bribes2, lloyd2023emergent, vote-buying, bribes5}, motivating our model of a bribing adversary.

In this work, we study how to design protocols that incentivize participants to raise timely alerts despite bribery. 
We design and analyze protocols that economically guarantee that if a subset of the nodes fail to alert, others are incentivized to alert instead.

We formalize the \emph{alerting problem}~(\S\ref{sec:model}) as a cryptoeconomic game involving a set of $n$ rational nodes, a smart contract, and a rational adversary.
The smart contract is deployed by a protocol operator who provisions it with a constant budget.
Nodes deposit tokens with the contract in the form of collateral, of which they lose a constant share if they fail to alert when needed.
The adversary, who gains a profit from a required on-chain alert not occurring, can offer bribes to nodes to deviate from the protocol and not alert. 
Our central design goal is to maximize \emph{bribery resistance}: the amount an adversary must pay in bribes to prevent any alert from being published on-chain.

We begin by considering a straightforward protocol~(\S\ref{sec:publishing_protocol_example}) where alerters split a fixed reward and non-alerters lose a constant penalty (deposit).
We show it achieves~$\Theta(n)$ bribery resistance. 
If, on the other hand,
 we use penalties of non-alerters to reward alerters, the total per-node reward can grow linearly in $n$. 
 As there are $n$ nodes, this approach promises \emph{quadratic}, i.e.,~$\Omega(n^2)$, bribery resistance. 

We in fact offer a very simple proof of an $O(n^2)$ upper bound~(\S\ref{sec:max_bribery_resistance})~on bribery resistance for any protocol with constant operator budget and constant penalties for deviating.

We present the \emph{simultaneous alerting game}~(\S\ref{sec:base_game}), a Stackelberg game that asymptotically achieves this quadratic upper bound.
In this game, all nodes act simultaneously after observing the adversary's bribe offers.
If the bribe offered to a node is higher than the maximum reward she can get, then her dominant strategy is not to alert.
And if the bribe is lower than the penalty for not alerting, her dominant strategy is to alert.
Surprisingly, we prove that in the range of bribe amounts between these two extremes, no symmetric pure equilibrium exists.
We thus turn our attention to mixed-strategy equilibria. We show that in any such equilibrium,
if the adversary's gain from a successful attack is less than the amount required to bribe all nodes with the maximum possible reward each---a total quadratic in the number of nodes---then her expected utility from the bribe is non-positive.
Thus, the simultaneous alerting game achieves~$\Theta(n^2)$ bribery resistance.

We present two protocols~(\S\ref{sec:realizing_base}) that realize the simultaneous game under different assumptions.
The \emph{Lockstep alerting protocol} achieves constant-time finality with a short constant alerting window. 
It relies on a strict synchrony assumption, namely that nodes' transactions are added to the blockchain after an exact constant time.
It publishes no transactions when no alert is needed and up to~$n$ transactions when alerting; it achieves~$\Theta(n^2)$ bribery resistance. 

The strict synchrony assumption, however, does not always hold~\cite{tx-done-yet,finalization_time_blockchains}. 
To relax it, we employ a commit-reveal scheme.
A purely cryptographic approach would enable nodes to hide their actions.
But a node could \textit{collude} with an adversary to reveal and prove its commitment prematurely. Such early revelation could be forwarded to other nodes in support of a bribery strategy that deviates from the simultaneous game. 
We therefore present a second protocol that makes use of \emph{trusted hardware}.
Its alerting window is divided into a \emph{commit} phase and a \emph{reveal} phase.
Nodes commit during the commit phase using the trusted hardware, which uses a proof of publication (PoP) primitive~\cite{Ekiden} to ensure they cannot reveal their action before a predefined time period elapses.
They then reveal their action during the reveal phase.
This protocol requires~$\Theta(n)$ transactions, i.e.,~$\Theta(n)$ storage on the blockchain, and also achieves~$\Theta(n^2)$ bribery resistance. 

To reduce the on-chain overhead and avoid both the strict synchrony and trusted hardware assumptions, we consider the alerting protocol proposed in the Chainlink v2 whitepaper~\cite{chainlink2021whitepaper}. 
The authors sketch a sequential mechanism, where nodes take turns raising alerts.
We formalize this \emph{sequential alerting protocol}~(\S\ref{sec:sequential_alerting_protocol}).
It sequences the nodes in~$n$ distinct time slots, removing the need to enforce simultaneity.
Each node can alert only during its assigned time slot, after observing the actions of all nodes in earlier slots.
Once a node alerts, the protocol ends.
This calls for a different reward structure: If a node alerts, it receives as a reward the cumulative penalties of all nodes assigned to earlier slots that did not alert. 
We analyze the induced game by backward induction and show that it admits a \emph{Subgame-Perfect Nash Equilibrium (SPNE)} in which, to suppress all alerts, the adversary must spend a budget quadratic in the number of nodes.
The sequential protocol requires no transactions when an alert is not needed and a single transaction otherwise, therefore constant storage on-chain, and achieves~$\Theta(n^2)$ bribery resistance, albeit only half of the quadratic budget enforced by the simultaneous alerting game. 

Our analysis of the sequential protocol additionally reveals that if an adversary can gain a profit by delaying an alert for~$m<n$ time slots rather than suppressing it entirely, then she needs to spend a budget quadratic in the number of slots~$\Theta(m^2)$.  

A qualitative comparison of these three protocols, highlighting their assumptions, alert latency, transaction requirements, and bribery resistance, is summarized in Table~\ref{tab:protocol-comparison-simple}.

In summary, our main contributions are: (1) formalization of the alerting problem and bribery resistance;
(2) proof that the maximum bribery resistance a protocol can achieve is~$\Theta(n^2)$;
(3) design of the simultaneous alerting game that asymptotically achieves this bound; and 
(4) three protocols that operate under different assumptions and tradeoffs, all achieving asymptotically-optimal bribery resistance.


\begin{table}[t]
\centering
\renewcommand{\arraystretch}{1.3}
\setlength{\tabcolsep}{4pt}
\small   

\begin{tabular}{|p{0.18\columnwidth}|p{0.15\columnwidth}|p{0.12\columnwidth}|p{0.1\columnwidth}|p{0.1\columnwidth}|p{0.15\columnwidth}|}
\hline
\textbf{Mechanism} & \textbf{Assump-tions} & \textbf{Alert latency} & \textbf{Tx for $\noalert$} & \textbf{Tx for $\alert$} & \textbf{Bribe cost} \\
\hline
\textbf{Lockstep} & Constant time to write txs on-chain & $\Theta(1)$ & 0 & $\Theta(n)$ & $\penaltySlash  {n(n-1)}$ \\
\hline
\textbf{Trusted-Hardware} & Trusted hardware & $\Theta(1)$ & $\Theta(n)$ & $\Theta(n)$ & $\penaltySlash  n(n-1)$ \\
\hline
\textbf{Sequential} & - & $\Theta(n)$ & 0 & 1 & $\penaltySlash \frac{n(n-1)}{2}$ \\
\hline
\end{tabular}
\vspace{3em}
\caption{Qualitative comparison of the alerting protocols. $n$~denotes the number of nodes and $\penaltySlash$ denotes the penalty.}
\setlength{\belowcaptionskip}{-10pt}

\label{tab:protocol-comparison-simple}
\end{table}

\section{Related Work}
\label{sec:related_work}
\input{related_work.tex}

\section{Model}
\label{sec:model}
To reason about the problem at hand, we define the \emph{alerting model}~(\S\ref{sec:participants_communication}) and our goal~(\S\ref{sec:goals}).

\subsection{Participants and Communication}
\label{sec:participants_communication}
The system comprises a set~$\mathcal{N}$ of~$n$ \emph{nodes}, a \emph{protocol operator}, a \emph{blockchain}, and an \emph{adversary}. 
The blockchain is a trusted party that maintains a state comprising a token balance for each node and supports the execution of arbitrary stateful programs called~\emph{smart contracts}.
The smart contract can access the state of the blockchain, specifically which block it is on.
The protocol operator implements the \emph{alerting protocol} as a smart contract, \AlertsCon, and publishes it on the blockchain.

Initially, before the execution of the alerting protocol, each node has some token balance.
Nodes interact with the smart contract by issuing \emph{transactions} that update its state.
To participate in the protocol, each node locks a certain amount of its tokens in \AlertsCon as collateral by issuing a \emph{stake} transaction. 
The contract can \emph{slash} (discard) the node's stake with a penalty~$\penaltySlash > 0$ that the node loses if she fails to follow the protocol.
It can also reward nodes that follow the protocol with tokens that the contract adds to their balance.

Time progresses in steps.
In each step, participants receive messages from previous steps, execute local computations, update their state, and send messages.
All participants can observe the current state of the blockchain. 
Communication between nodes and the blockchain is synchronous:
Nodes can issue transactions that are added to the blockchain within a bounded time~$\NetworkWrite$, which is known to all participants~\cite{bitcoin-backbone, prism}.
We will examine two communication models:
One where transactions are appended after exactly~$\NetworkWrite$ steps, and another where transactions are appended within at most~$\NetworkWrite$ steps.
The first model, the \emph{lockstep model}, captures a strong synchrony assumption that does not typically hold~\cite{tx-done-yet,finalization_time_blockchains}, we use it as a stepping stone to design our protocols,
while the second, the \emph{bounded delay model}, captures a standard synchrony assumption~\cite{dist-sys}.

Nodes continuously monitor both the blockchain and events outside the chain.
The protocol operator specifies a set of conditions that define when an alert is needed, e.g., a \emph{prediction market} being resolved with an incorrect outcome~\cite{uma-oo-overview} or an \emph{oracle} price update that is missing~\cite{data-streams-cl}.
When an event meets these conditions, nodes have the chance to alert during a predefined \emph{alerting period} by sending an \emph{alert transaction} to \AlertsCon.

An \emph{alerting protocol} is a protocol that an operator implements in a smart contract, \AlertsCon.
The protocol incentivizes a set of participants, who stake tokens with the contract, to monitor the blockchain and off-chain events and to alert when events meet a specified set of conditions.

We consider alerts that are \emph{adjudicable}, i.e., there exists some mechanism (e.g., on-chain  logic~\cite{ethereum-rewards-penalties} or an external committee~\cite{chainlink2023decentralized-model}) that can verify the correctness of an alert, and we model this mechanism as a function~$\checkreport{\data}$ that the contract operator publishes with the smart contract.    
The function~$\checkreport{\data}$ takes as input data about off-chain events and returns \tr if the conditions for alerting are satisfied, and thus nodes should alert, and \fa otherwise.

Since the smart contract cannot directly access off-chain data, a node that detects an event that meets the alerting conditions must submit a proof~$\data$ along with the alert transaction such that~$\checkreport{\data} = \tr$.
 \AlertsCon verifies the proof upon receiving the transaction.
This proof may be the data itself or a cryptographic proof (e.g., a zero-knowledge proof), depending on the usecase.
If the contract verifies the proof successfully, the alert is considered \emph{valid}. 
In this case, the contract may invoke additional functionality that the operator defines to respond to the alert, and may slash or reward nodes based on their actions in the alerting protocol.
If the alerting period ends without a valid alert, the contract takes no additional actions.
The response to an alert is application-specific (e.g., shutting down a service or triggering wider security checks) and is abstracted away in this work by having the contract~\AlertsCon invoke an \textbf{Alert} function.

We study protocols where the amount of stake a node locks is independent of the number of nodes~$n$, but above some minimum threshold and with a constant amount of tokens~$\operatorCost \geq 0$ which the protocol operator funds the contract with before the protocol starts called the \emph{operator cost} (as in, e.g., Ethereum~\cite{ethereum-staking} and UMA~\cite{umaDVM}). 
Since each node's stake does not scale with~$n$, any slashing is upper bounded by that stake, i.e., a constant amount.
Hence, we focus on protocols that use a constant per-node penalty for deviating.
\begin{definition}[($\operatorCost, \penaltySlash$)-Alerting Protocol]
  A \emph{($\operatorCost, \penaltySlash$)-alerting protocol} is an alerting protocol where the operator funds the contract with a constant initial budget~$\operatorCost \geq 0$, and the contract penalizes nodes that deviate from the protocol with a constant penalty~$\penaltySlash > 0$ each.
\end{definition}

We consider a probabilistic polynomial time (PPT) adversary, with standard cryptographic assumptions.
She can observe the blockchain and offer \emph{bribes} to nodes for deviating from the protocol. 
In practice, the adversary can realize such bribes either on-chain (e.g., via a smart contract) or off-chain (e.g., via a different chain) and can condition the bribe on the node's on-chain actions, or the lack thereof.
The adversary offers each node~$i$ a bribe~$\beta_i \ge 0$ that pays out only if the node chooses not to alert.
Each node's bribe is private information known only to the node and the adversary.

If the adversary successfully prevents all alerts, she gains a positive profit~$\advGain$.
Nodes are \emph{rational}, aiming to maximize their profit.
They follow the protocol, unless deviating from it results in a higher profit.

\subsection{Alerting Protocols and Bribery Resistance}
\label{sec:goals}
We are interested in designing alerting protocols that impose a high cost on an adversary trying to suppress alerts.
\begin{definition}[$\sigma$-Bribery Resistance]
An alerting protocol achieves \emph{$\sigma$-bribery resistance} if it guarantees that an alert gets published on-chain whenever the adversary's total budget for bribes is at most~$\sigma$.
\end{definition}




\begin{note2}[When No Alert is Needed]
We focus on the case when an alert is needed, i.e., when there exists~$\data$ such that $\checkreport{\data} = \tr$.
When $\checkreport{\data} = \fa$ the conditions for alerting are not met and there is no event that the smart contract needs to be notified about. 
This is the desired state of the system. 
Unless a node is offered a bribe to do so, it is never in its interest to alert when $\checkreport{\data} = \fa$ since, in this case, alerting yields no reward while incurring the cost of the transaction. 
For an adversary, there is similarly no benefit in inducing alerts when not needed since this does not change the system outcome. 
Therefore, in the rest of the paper, we focus on the case when an alert is needed, i.e., when~$\checkreport{\data} = \tr$.
\end{note2}

\section{Burned-Penalty Protocol Example}
\label{sec:publishing_protocol_example}

As a first example, we present a straightforward alerting protocol.
This protocol operates in the {lockstep network model}, where all transactions sent by nodes at a timestep are appended to the blockchain after exactly~$\NetworkWrite$ steps.
This strong synchrony assumption ensures that all nodes' alert transactions, if sent, appear on-chain simultaneously and nodes are not aware of other nodes' actions during the alerting period.

The alerting period lasts for a single time step.
All nodes can alert during this time step.
After~$\NetworkWrite$ steps following the end of the alerting period, the contract verifies whether at least one valid alert was published on-chain.
The contract implements the following reward and penalty structure:
It rewards all nodes that sent a valid alert with a share of the operator's budget~$\operatorCost$, divided equally among them.
And slashes all nodes that did not alert with a constant penalty~$\penaltySlash$ each, which is \emph{burned} (removed from the system).

The protocol induces a game between the adversary and the nodes.
The adversary can offer each node~$i$ a private bribe~$\beta_i$ to not alert.
Each node then decides whether to alert or accept the bribe, aiming to maximize its utility.

We analyze this protocol's bribery resistance.
Notice that the maximum reward a node can receive by alerting is~$\operatorCost$, which occurs when it is the only node that alerts.
If a node accepts a bribe~$\beta_i$ and does not alert, but at least one other node alerts, then the node is penalized with~$\penaltySlash$.
In this case, the node's profit is~$\beta_i - \penaltySlash$.
If~$\beta_i >  \operatorCost + \penaltySlash$, the node's dominant strategy is to accept the bribe and not alert, since this guarantees a profit larger than the maximum reward it could receive by alerting, even accounting for the potential penalty.

Therefore, for an adversary to successfully convince a node not to alert, it is sufficient to bribe her with~$\operatorCost + \penaltySlash$.
If the adversary bribes all~$n$ nodes with this amount, no node will alert, and the adversary succeeds in suppressing all alerts.
This requires a total budget of~$n \cdot (\operatorCost + \penaltySlash)$.
Therefore, this protocol achieves at most linear bribery resistance, i.e.,~$\Theta(n)$.

\section{Maximum Bribery Resistance}
\label{sec:max_bribery_resistance} 

We show that the maximum bribery resistance that a ($\operatorCost, \penaltySlash$)-protocol can achieve is quadratic in the number of nodes.
The Burned-Penalty protocol from Section~\ref{sec:publishing_protocol_example} rewards nodes for alerting, but the total reward is a constant value equal to the operator's cost~$\operatorCost$.
Instead, we consider protocols that use the slashed value from non-alerting nodes to pay rewards to alerting nodes.
This allows the total available reward to scale with the number of nodes, increasing the cost for an adversary to bribe all nodes not to alert.

The key insight is that although protocols may implement different reward and penalty structures, if each node's staked amount is independent of the number of nodes~$n$, then the maximum penalty a node can incur is bounded by a constant, which limits how much reward can be distributed.

\begin{claim}
[Maximum bribery resistance]
  The maximum bribery resistance that a protocol with~$n$ nodes and constant operator cost~$\operatorCost$ where each node stakes a constant amount of tokens~$\penaltySlash$ can achieve is~$\penaltySlash\, n^2 + \operatorCost \, n$.
\end{claim}

\begin{proof}
  Consider an alerting protocol with~$n$ nodes where each node stakes~$\penaltySlash$ tokens and the operator funds the contract with~$\operatorCost$ tokens.
  Consider a node~$i$. 
  The maximum amount available in the system comes from all other~$n-1$ nodes, which could be slashed, each incurring a penalty of at most~$\penaltySlash$.
  In this case, the maximum reward node~$i$ can receive by alerting is~$\penaltySlash (n-1) + \operatorCost$.

  Since the stake amount is not dependent on~$n$, the maximum penalty a node can incur if it deviates from the protocol is constant and bounded by her stake~$\penaltySlash$.
  Therefore, for an adversary to convince a node to deviate, it is sufficient to offer her a bribe of~$\penaltySlash \, n + \operatorCost$.
  Since this holds for any node, the adversary can convince all nodes to deviate by offering a total bribe of~$\penaltySlash \, n^2 + \operatorCost \, n$. 
\end{proof}

In a~($\operatorCost, \penaltySlash$)-protocol, the penalty for deviating from the protocol is not bounded by, but exactly equal to~$\penaltySlash$.
Therefore, any~($\operatorCost, \penaltySlash$)-protocol with~$n$ nodes has the same upper bound on bribery resistance.
In the following sections, we design~($\operatorCost, \penaltySlash$)-protocols that asymptotically achieve this upper bound, and thus are asymptotically optimal.

Since the operator cost~$\operatorCost$ only adds a linear term to the maximum bribery resistance, in the rest of the paper, for simplicity, we set~$\operatorCost = 0$ and focus on ($0, \penaltySlash$)-protocols.
We refer to these simply as \emph{$\penaltySlash$-alerting protocols}.
We present the full analysis with non-zero operator cost in Appendix~\ref{app:non_zero_operator_cost} and show that it does not change our results asymptotically.

\section{Simultaneous Alerting Game}
\label{sec:base_game}
Our goal is to design~$\penaltySlash$-alerting protocols that achieve asymptotically optimal bribery resistance of~$\Theta(n^2)$.
To this end, we first design a simple abstract game that runs in constant time called the \emph{simultaneous alerting game}, and show it achieves quadratic bribery resistance.
We start by defining the game as a Stackelberg game~(\S\ref{sec:simultaneous_alerting_game}).
We characterize its equilibria under different ranges of the adversary's gain from a successful attack.
First, we identify conditions under which either alerting or not alerting is a dominant strategy, leading to a symmetric pure equilibrium in each case~(\S\ref{sec:pure-strategy-nash-equilibria}).
We then show that in the intermediate range of bribes no symmetric pure equilibrium exists~(\S\ref{sec:pure-strategy-nash-equilibria}). 
Finally, we extend the analysis to mixed-strategy equilibria and study the adversary's optimal bribe choice and expected utility~(\S\ref{sec:mixed-strategy-equilibria}).

\subsection{Game Definition}
\label{sec:simultaneous_alerting_game}
We define the \emph{simultaneous alerting game}, or in short the \emph{alerting game}, an abstract game that asymptotically achieves the maximum bribery resistance. 
Two types of players are involved: the set of nodes $\mathcal{N}$ and an adversary.
The game proceeds in two stages.
First, the adversary offers bribes by choosing a vector $(\beta_1,\dots,\beta_n)$ specifying how much each node would be paid for not alerting, regardless of the game outcome.

For our primary analysis, we assume the adversary pays the bribe to any node that chooses \noalert, regardless of the game's final outcome. 
An alternative strategy for the adversary would be to offer a \textit{conditional bribe}, where a node that chooses \noalert receives the bribe~$\beta_i$ only if the attack succeeds, i.e., if no node alerts.
We discuss this modification in Appendix~\ref{app:conditional_bribes} and show that it does not change our results.

The bribes are private information between the adversary and each node, meaning that no other node knows the bribe offered to another node.
Then, given these offers, each node $i \in \mathcal{N}$ decides whether to accept the bribe and skip alerting or to alert and not get the bribe, choosing an action $a_i \in \{\alert, \noalert\}$.
Afterwards, the adversary pays the bribe only to nodes whose on-chain behavior matches the \noalert action.
The protocol rewards nodes that follow the protocol.

In an execution of the game, denote by~$\alerters$ the set of players who choose $\alert$.
If~$\alerters$ is not empty, all nodes~$j \notin \alerters$ pay a penalty of~$\penaltySlash$
and all nodes~$i \in \alerters$ share the slashed value equally as a reward.
Each alerter gets~$ \penaltySlash \frac{ (n - |\alerters|) }{|\alerters|}$.
Node~$i$ may choose not to alert in exchange for a bribe of~$\beta_i$.
A node that alerts receives no bribe.
In summary, the utility of each node~$i$ is its payoff:
\begin{equation}
\label{eq:sim_payoff}
\rho_i(a_i,\beta_i,\alerters)
= 
\begin{cases}
    \penaltySlash\frac{  (n-|\alerters|)}{|\alerters|}, 
      &  a_i = \alert\\  

    \beta_i - \penaltySlash, 
      &  a_i = \noalert, \alerters \neq \emptyset \\

    \beta_i, 
      & \alerters = \emptyset .
\end{cases}
\end{equation}

For the adversary, if none of the nodes alert, she gains~$\advGain$ but pays all bribes.
If a node alerts, the adversary gains nothing but pays the bribes to nodes that did not alert.
Thus, the adversary's utility is her payoff:
\begin{equation}
\label{eq:sim_payoff_adv}
\rho_{\text{adv}}(\beta_1,\dots,\beta_n,\alerters) = 
\begin{cases}
\advGain - \sum\limits_{i \in \mathcal{N}} \beta_i, & \alerters = \emptyset\\

- \sum\limits_{i \in \mathcal{N} \setminus \alerters} \beta_i, & \alerters \neq \emptyset 
\end{cases}
\end{equation}

 

\subsection{Pure-Strategy Nash Equilibria}
\label{sec:pure-strategy-nash-equilibria}
We analyze the pure-strategy Nash equilibria (NE) of the simultaneous alerting game.
We show that there are two cases where a pure dominant strategy exists.
If the bribe is sufficiently high, then each node's dominant strategy is to choose $\noalert$.
And if the bribe is sufficiently low, then each node's dominant strategy is to choose $\alert$.
In the latter case, the adversary's optimal strategy is to not offer any bribe.

\begin{claim}[{\noalert} dominance]
\label{clm:noalert-dominant}
In the simultaneous alerting game, if the adversary's offered bribe for a node~$i$ is~$\beta_i\ge \penaltySlash(n-1)$, then the action~$\noalert$ weakly dominates $\alert$ for node~$i$. 
And $\noalert$ is strictly dominant if $\beta_i>\penaltySlash(n-1)$.
\end{claim}

\begin{proof}
Consider a node~$i$ and let~${y\in\{0,1,\ldots,n-1\}}$ denote the number of other alerters in an execution of the game.
If~$i$ alerts while~$y$ others alert, i.e.,~${|\alerters|=y+1}$, its payoff is~$\penaltySlash\frac{n-y-1}{y+1}$.
Otherwise, if~$i$ does not alert and someone alerts, its payoff is $\beta_i-\penaltySlash$.
Finally, if no one alerts, i.e.,~$y=0$, choosing $\noalert$ earns node~$i$ the offered bribe~$\beta_i$ while a single alerter earns $\penaltySlash(n-1)$.

We divide into two cases based on the value of $y$.

\begin{enumerate}
\item \textbf{$y=0$.}
Comparing $\noalert$ and $\alert$ gives ${u_i(\noalert,\beta_i, \alerters=\varnothing)=\beta_i}$ and ${u_i(\alert, \beta_i, \alerters=\{i\})=\penaltySlash(n-1)}$.
By assumption, $\beta_i\ge \penaltySlash(n-1)$, so $\noalert$ is at least as good, and strictly better if $\beta_i>\penaltySlash(n-1)$.

\item \textbf{ $y\ge 1$.}
Again we compare the utility of the two actions.
$u_i(\noalert, \beta_i, \alerters)=\beta_i-\penaltySlash$, while~$u_i(\alert, \beta_i, \alerters)=\penaltySlash\,\frac{n-1-y}{y+1}$.
The function~$f(y):=\frac{n-y-1}{y+1}$ is decreasing in~$y$ for
$y\ge 1$, so its maximum over $\{1,\ldots,n-1\}$ is at $y=1$ with~$f(1)=\frac{n-2}{2}$. 
Thus it suffices to show~$
\beta_i-\penaltySlash \;\ge\; \penaltySlash\, \frac{n-2}{2}.
$

This holds whenever $\beta_i\ge \penaltySlash \, \frac{n}{2}$, and in particular under our assumption $\beta_i\ge \penaltySlash (n-1)$ because ${\penaltySlash(n-1)\ge \penaltySlash \, \tfrac{n}{2}}$ for all ${n\ge 2}$. 
Therefore, for every~${y\ge 1}$,~${u_i(\noalert, \alerters)\ge u_i(\alert, \alerters)}$, with strict inequality when ${\beta_i>\penaltySlash(n-1)}$.
\end{enumerate}
Combining the two cases, $\noalert$ dominates $\alert$ for all nodes~$i$ whose~$\beta_i\ge \penaltySlash(n-1)$, and is strictly dominant when ${\beta_i>\penaltySlash(n-1)}$.
\end{proof}

\begin{corollary}[Profitable bribery when~$\advGain > \penaltySlash \, n(n-1)$]
\label{cor:profitable-bribery}
In the full two-stage alerting game, if the adversary's gain from a successful attack is~${\advGain > \penaltySlash \, n(n-1)}$, then there exists a profitable bribery strategy for the adversary that induces a unique pure-strategy equilibrium in the second stage of the game where all nodes choose $\noalert$.
\end{corollary}
\begin{proof}
By Claim~\ref{clm:noalert-dominant}, for every node $i$, if $\beta_i > \penaltySlash \, (n-1)$ then $\noalert$  dominates $\alert$.
Hence, if the adversary offers $\beta_i > \penaltySlash(n-1)$ to all nodes, the nodes' unique best-response profile is that all nodes choose $\noalert$. 
The adversary's utility at these $\beta_i$ values and nodes' response is $u_{\mathrm{adv}} = \advGain - \sum_{i=1}^{n} \beta_i$ if she chooses to bribe the nodes, and $u_{\mathrm{adv}} = 0$ if she does not. 
Thus, when the adversary's gain from a successful attack is~${\advGain > \penaltySlash \, n(n-1)}$, then bribery is profitable.
\end{proof}

We conclude that when the briber's gain from a successful attack is over~$\penaltySlash n(n-1)$, she can ensure that all nodes play $\noalert$ and achieve a positive utility.
We now consider the opposite case where the bribe is too low.

\begin{claim}[\alert\ is dominant when $\beta_i  \leq \penaltySlash$]
\label{clm:alert-dominant}
\lemmaAlertDominance
\end{claim}




\begin{corollary}[Bribery is unprofitable when $\advGain < \penaltySlash \, n$]
  \label{cor:unprofitable-bribery}
\corUnprofitableBribery
\end{corollary}


The proofs of Claim~\ref{clm:alert-dominant} and Corollary~\ref{cor:unprofitable-bribery} are similar to those of Claim~\ref{clm:noalert-dominant} and Corollary~\ref{cor:profitable-bribery}, respectively, and can be found in Appendix~\ref{app:alert-dominance}.

We showed that if the offered bribe~$\beta_i\leq\penaltySlash$, or~$\beta_i\geq \penaltySlash(n-1)$, then there exists a dominant strategy for the node~$i$.
However, we now show that if the bribes offered~$\penaltySlash<\beta_i<\penaltySlash(n-1)$, then no symmetric pure NE exists.

\begin{lemma}[No symmetric pure equilibrium]
\label{lem:no-symm-pure}
 Consider the second phase of the simultaneous alerting game 
with~$n\ge2$ and a~$\penaltySlash>0$.
If there exists a node~$i$ whose bribe~$\beta_i$ satisfies 
\begin{equation}
\label{eq:interior}
\penaltySlash \;<\; \beta_i \;<\; \penaltySlash(n-1),
\end{equation}
then there exists no symmetric pure NE.
\end{lemma}

\begin{proof}
There are only two symmetric pure profiles:
\begin{enumerate}
  \item \emph{All \alert: $\alerters=\mathcal N$.}
Then each node's payoff is~$0$ (Equation~\ref{eq:sim_payoff}).
A unilateral deviation by node~$i$ to \noalert yields payoff~${\beta_i-\penaltySlash}$ (Equation~\ref{eq:sim_payoff}) since~${\alerters\setminus\{i\}\neq\emptyset}$.
By the left inequality in Equation~\ref{eq:interior},~$\beta_i-\penaltySlash>0$, so deviating is strictly profitable.
Hence all~\alert is not a NE.

  \item  \emph{All \noalert: $\alerters=\emptyset$.}
Node~$i$'s payoff is~$\beta_i$.
A unilateral deviation by node~$i$ to~$\alert$ yields payoff~$\penaltySlash\frac{|\,\mathcal N\setminus\{i\}\,|}{|\,\{i\}\,|}=\penaltySlash(n-1)$.
By the right inequality in Equation~\ref{eq:interior},~$\penaltySlash(n-1)>\beta_i$, so deviating is strictly profitable.
Hence, all~\noalert is not a NE.
\end{enumerate}

Neither of the two symmetric pure profiles is~a~NE.
\end{proof}


\subsection{Mixed-Strategy Nash Equilibria}
\label{sec:mixed-strategy-equilibria}
By Lemma~\ref{lem:no-symm-pure}, whenever a bribe satisfies~${\penaltySlash < \beta_i < \penaltySlash(n-1)}$, no symmetric pure Nash equilibrium exists.
Since a node that is offered such a bribe has no dominant strategy, it may choose its action with some probability. 
We study the existence of mixed-strategy equilibria in this interior range of bribes, where the strategy of a node is her probability of choosing $\alert$.
In the next lemma, we give a lower bound on the adversary's \emph{expected} total bribe outlay that holds for any bribe vector and any strategy profile of the nodes.

\begin{claim}[Lower bound on expected bribe]
\label{lem:bribe-lower-bound}
\lemmaLowerBoundBribery
\begin{equation}
\label{eq:lower-bound-bribes}
\sum_{i=1}^n \beta_i\, q_i
\;\ge\;
\penaltySlash\, n(n-1)\, Q.
\end{equation}
\end{claim}
We present a sketch of the proof, full details can be found in Appendix~\ref{app:mixed-strategy-equilibria}

\begin{proof}
Consider a node~$i$ and let~$Y_{-i}$ be the number of other alerters (whether or not~$i$ alerts). 
Let~${Q_{-i}:=\Pr(Y_{-i}=0)=\prod_{j\ne i} q_j}$ be the probability that other nodes do not alert.
If~$i$ alerts while~${Y_{-i}=y}$, its payoff is~${\penaltySlash(\frac{n}{1+y}-1)}$, so
\[
\mathbb{E}[\rho_i(a_i = \alert,\beta_i,\alerters)]
=\penaltySlash(n\,\mathbb{E}\!\Big[\frac{1}{1+Y_{-i}}\Big]-1).
\]
We denote this expected payoff by $U_i^{\alert}$.
If $i$ chooses $\noalert$, it always receives $\beta_i$ and pays the penalty iff someone else alerts:
\[
\mathbb{E}[\rho_i(a_i = \noalert,\beta_i,\alerters)]
= \beta_i + \penaltySlash\,Q_{-i} - \penaltySlash.
\]
We denote this expected payoff by $U_i^{\noalert}(\beta_i)$.

Notice that~$U_i^{\alert}$ is independent of~$\beta_i$ while~$U_i^{\noalert}(\beta_i)$ is strictly increasing in~$\beta_i$.
In equilibrium, for each node~$i$, its utility from alerting and not alerting must be equal. 
I.e.,~${U_i^{\alert} = U_i^{\noalert}(\beta_i)}$.
Otherwise, a node would deviate to the action with the higher expected utility.
Thus, in equilibrium, it holds that 
\[
\penaltySlash(n\,\mathbb{E}\!\Big[\frac{1}{1+Y_{-i}}\Big]-1) = \beta_i + \penaltySlash\,Q_{-i} - \penaltySlash.
\]

From this equation, and since $\tfrac{1}{1+y}\ge 1$ when $y=0$ and $\tfrac{1}{1+y}\ge0$ otherwise, we get that
$
\beta_i \;\ge\; \penaltySlash\,(n-1)\,Q_{-i}$.
Multiplying both sides by~$q_i$ and using~${q_i Q_{-i}=\prod_{j=1}^n q_j=Q}$ yields~$\beta_i q_i \ge \penaltySlash\,(n-1)\,Q$ for all such $i$.
Summing over $i=1,\ldots,n$, we get
\[
\sum_{i=1}^n \beta_i q_i \;\ge\; \penaltySlash\,(n-1)\,\sum_{i=1}^n Q
\;=\; \penaltySlash\,n(n-1)\,Q,
\]
concluding the proof.
\end{proof}

In the boundary cases where a dominant strategy exists, if all bribes are at least~$\penaltySlash(n-1)$, then all nodes choose $\noalert$, i.e., $q_i=1$ for all $i$ and $Q=1$.
Similarly, if all bribes are at most~$\penaltySlash$, then all nodes choose $\alert$, i.e., $q_i=0$ for all $i$ and $Q=0$.
So inequality~\ref{eq:lower-bound-bribes}~holds.

To analyze the adversary's optimal strategy, we consider her expected payoff.
The adversary's expected payoff is her gain from a successful attack times the probability that no node alerts, minus the expected total bribe paid to the nodes.
I.e.,~${\mathbb{E}[\rho_{\text{adv}}]=\advGain\,Q-\sum_i \beta_i q_i}$.
By Lemma~\ref{lem:bribe-lower-bound}, we get the upper bound~${\mathbb{E}[\rho_{\text{adv}}]\le (\advGain-\penaltySlash\,n(n-1))\,Q}$,
which is nonpositive for~$\advGain\le \penaltySlash\,n(n-1)$.

Therefore, if the adversary's gain from a successful attack is less than~$\penaltySlash\,n(n-1)$, then her expected payoff is nonpositive for any bribe vector and any mixed-strategy equilibrium.
So her optimal strategy is to not bribe any node, yielding zero expected payoff.

\section{Realizing the Simultaneous Game}
\label{sec:realizing_base}
We present two alerting protocols that realize the simultaneous alerting game.
The game requires all nodes to choose their actions simultaneously, i.e., no node can observe other nodes' actions before choosing its own, and neither can the adversary.
{We identify three properties of alerting protocols and show that they are sufficient for a constant-time~$\penaltySlash$-alerting protocol to satisfy to implement the game~(\S\ref{sec:realizing_base_requirements}).}
To ensure these properties in practice, we design two protocols that work under a different sets of assumptions.
The first protocol, \emph{Lockstep alerting}~(\S\ref{sec:synchronous_protocol}), works under a strong network assumption.
We relax this assumption in the second protocol, \emph{TEE-based alerting}~(\S\ref{sec:hardware_based_protocol}), by using trusted hardware components to implement a timed commitment scheme.

\subsection{Sufficient Conditions}
\label{sec:realizing_base_requirements}
We identify three properties that if a constant-time~$\penaltySlash$-alerting protocol satisfies, then it implements the simultaneous alerting game.

\begin{definition}[Alert Capability]\label{req:alert_capability}
Each node is able to raise an alert before the alerting period ends.
\end{definition}

\begin{definition}[Deniability]\label{req:private_decisions_no_evidence}
  For any execution~$E_0$ in which node~$i$ alerts (respectively, does not alert), 
  there exists another execution~$E_1$ in which node~$i$ does not alert (respectively, alerts) such that no node~$j \neq i$ and no adversary can distinguish between~$E_0$ and~$E_1$ before {choosing their own actions}. 
\end{definition}
{Note that if a protocol achieves Deniability, in particular, node~$i$ cannot generate any message that would help others distinguish between the two executions before time~$t_0+1$.}

\begin{definition}[Reward Distribution]\label{req:reward_distribution}
  The protocol distributes all slashed penalties symmetrically to the nodes that alerted.
\end{definition}

We prove that the three properties are sufficient.
\begin{proposition}
\label{prop:sufficiency}
Any constant-time $\penaltySlash$-alerting protocol that satisfies the properties of Alert Capability, Deniability, and Reward Distribution implements the simultaneous alerting game.
\end{proposition}

To prove this proposition, we show a mapping between the constraints enforced by the properties and the game-theoretic elements of the simultaneous alerting game.
\begin{proof}
  Let~$P$ be a $\penaltySlash$-alerting protocol that runs in constant time and satisfies the three properties.
  Let~$n$ be the number of nodes participating in the protocol.
  First, the players in the game correspond to the nodes in the protocol and the adversary.
  Second, each node~$i$ can either follow the protocol and alert when needed, or do anything else which is considered as not alerting.
  By the Alert Capability property, each node can always alert before the alerting period ends.
  I.e.,  node~$i$'s possible actions are~$\alert$ or~$\noalert$ as in the simultaneous alerting game. 
  For the adversary, she can offer a bribe~$\beta_i\ge 0$ to each node~$i$.
  Third, by the Deniability property, no party can distinguish between an execution where a node~$i$ alerts and another execution where node~$i$ does not alert before the alerting period ends.
  Therefore, no party can observe other nodes' actions before the alerting period ends.
  Thus, each node~$i$ chooses its action~$a_i$ without knowledge of any other node's choice~$a_j$.
  Similarly, the adversary offers bribes without prior knowledge of nodes' actions.
  Finally, by the definition of a $\penaltySlash$-alerting protocol, if a node deviates and does not alert when needed, it pays a penalty of~$\penaltySlash$.
  And by the Reward Distribution property, all nodes that alerted symmetrically share the slashed penalties.
  Therefore, if at least one node sends a valid alert, all non-alerting nodes pay a penalty of~$\penaltySlash$ and all alerting nodes share the slashed value equally as a reward.
  For the adversary, if no node alerts, she gains~$\advGain$ but pays all bribes.
  If at least one node alerts, she gains nothing but pays the bribes to non-alerting nodes.
  Thus, the node's and adversary's revenues in the protocol match the players' utilities defined in Equations~\ref{eq:sim_payoff} and~\ref{eq:sim_payoff_adv}, respectively.
  Since the protocol enforces the strategy space, information structure, and payoffs defined in the game, we conclude it implements the simultaneous alerting game.
\end{proof}

\subsection{Lockstep Alerting Protocol}
\label{sec:synchronous_protocol}
We present the \emph{Lockstep alerting} protocol (illustrated in Figure~\ref{fig:lockstep_alerting}) that realizes the simultaneous alerting game in the lockstep network model, where the time to write a transaction to the blockchain is a fixed constant~$\NetworkWrite$.
This network model is a strong variant of the standard synchronous network model that serves two purposes:
First, it allows us to illustrate the design of an alerting protocol that implements the simultaneous alerting game in a simple manner.
Second, it serves as a building block for a more complex protocol that works under relaxed assumptions and uses trusted-hardware components to implement the same knowledge constraints.

\begin{figure}[t]
\centering
\includegraphics[width=0.48\textwidth]{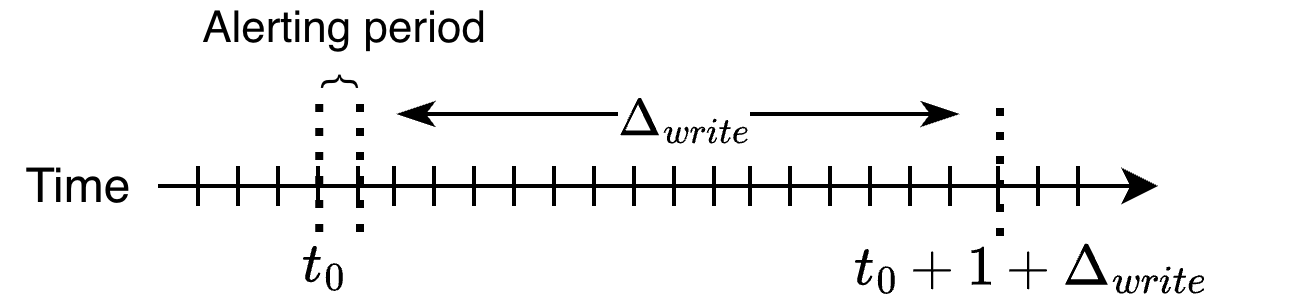}
\caption{Lockstep protocol: nodes alert during the same timestep.}
\label{fig:lockstep_alerting}
\end{figure}

\subsubsection{Protocol Description}
The Lockstep alerting protocol uses the smart contract~\AlertsCon to manage the alerting process.
In this protocol, the alerting period is of length one time step.
Nodes decide whether to alert or not prior to the alerting period and if a node decides to alert, it does so by sending an alert transaction during the alerting period.
An alert transaction includes a hash of a proof~$\Pi_\alert$ of the validity of her alert.

Once~$\NetworkWrite$ steps elapse after the alerting period, the smart contract checks if it received alert transactions.
Denote the block produced~$\NetworkWrite$ steps after the alerting period by~$B$.
By the lockstep assumption, any transaction sent during the alerting period is guaranteed to be included in block~$B$, and any transaction sent after that time is guaranteed not to be included in the same block.
If at least one node alerted, the protocol penalizes all non-alerting nodes with the constant penalty~$\penaltySlash$, and rewards all alerting nodes by sharing the slashed penalty symmetrically.
So if~$y < n$ nodes did not alert, each alerting node receives a reward of~$\penaltySlash \frac{y}{n-y}$.

\subsubsection{{Implementing the simultaneous alerting game.}}
We show that the  Lockstep alerting protocol implements the simultaneous alerting game.
We map the participants in the protocol and their actions to the players and actions in the game, and prove that the protocol satisfies Alert Capability, Deniability, and Reward Distribution. 

We first show that the information each player has when choosing its action is independent of other players' actions, as is the case in the simultaneous alerting game.

\begin{lemma}
\label{lem:lockstep_deniability}
The Lockstep alerting protocol achieves Deniability.
\end{lemma}

\begin{proof}
Let~$t_0$ be the start of the alerting period and~$B$ the block produced at~$t_0+1+\NetworkWrite$. 
First, notice that no node can alert strictly after time~$t_0$, since its alert transaction would not be included in block~$B$.
Consider a node~$i$ and an execution~$E_0$ of the alerting protocol where in~$E_0$ node~$i$ sends an alert transaction at time~$t_0$ (i.e., chooses~$\alert$).  
Consider a second execution~$E_1$ where all nodes and the adversary behave identically to~$E_0$ until time~$t_0+1$ except for node~$i$, which does not alert (i.e., chooses~$\noalert$).
Before time~$t_0$, both executions are identical, and any message sent at time~$t_0$ becomes available no earlier than time~$t_0+1$ by the lockstep assumption.
In particular, an alert transaction sent by node~$i$ in~$E_0$ becomes public on-chain only after~$\NetworkWrite$ steps, i.e., at time~$t_0+1+\NetworkWrite > t_0+1$.
Therefore, before the end of the alerting period at time~$t_0+1$, no party but node~$i$ itself can distinguish between executions~$E_0$ and~$E_1$.
In the same manner, for any execution~$E_0$ where node~$i$ does not alert, we can construct an execution~$E_1$ where node~$i$ alerts and no party can distinguish between the two executions before time~$t_0+1$, proving Deniability.
\end{proof}

\begin{theorem}
\label{thm:lockstep_implements_simultaneous}
The Lockstep alerting protocol is a $\penaltySlash$-alerting protocol that implements the simultaneous alerting game.
\end{theorem}

\begin{proof}
We prove that the game induced by the Lockstep alerting protocol is the simultaneous alerting game.
The players in the game correspond to the nodes in the protocol and a bribing adversary.
In an execution of the protocol, we map the action of node $i$ to~${a_i:= \alert}$ if node~$i$ sends an alert transaction during the alerting period, and~${a_i:=\noalert}$ otherwise.
A node can choose the action~$\noalert$ by simply not sending an alert transaction during the alerting period.
A node that decides to alert can send a valid alert transaction during the alerting period, ensuring (due to the lockstep assumption) that it is included in the corresponding block.
Therefore, the protocol achieves Alert Capability.

By Lemma~\ref{lem:lockstep_deniability}, the protocol achieves Deniability.
Finally, by definition of the protocol, if at least one node sends a valid alert transaction during the alerting period, all non-alerting nodes pay a penalty of~$\penaltySlash$, and all alerting nodes share the slashed value equally as a reward.
Therefore, the protocol achieves Reward Distribution.

To conclude, we proved that the Lockstep alerting protocol achieves the three sufficient properties and thus by Proposition~\ref{prop:sufficiency} induces the simultaneous alerting game.
\end{proof}


The Lockstep protocol requires zero transactions when no alert is needed, and~$n$ transactions otherwise.
The protocol runs in constant time, and the cost to bribe all nodes and cause its failure is~$\penaltySlash n (n-1)$ as shown in Lemma~\ref{lem:bribe-lower-bound}.

\begin{note2}[On the Lockstep Network Model]
  We require the strict synchrony assumption that writing takes exactly~$\NetworkWrite$ steps because the standard synchrony assumption does not specify a model, 
  e.g., probability distribution, over write times less than~$\NetworkWrite$. 
  One could consider a simple probabilistic model where a message sent at time~$t$ is included in the blockchain at time~$t+\NetworkWrite$ with probability~$1$, but messages sent at time~$t<t'<t+\NetworkWrite$ are included at time~$t+\NetworkWrite$ with probability~$0<p(t)$, decreasing with~$t$.
  It is still feasible to use the Lockstep alerting protocol in such a model, however, we get a tradeoff between the Alert Capability and Deniability properties.
\end{note2}




\subsection{Hardware-Based Alerting Protocol}
\label{sec:hardware_based_protocol}
The lockstep network model in which the Lockstep alerting protocol operates does not usually hold in practice~\cite{tx-done-yet,finalization_time_blockchains}.
To relax this assumption we propose a second protocol, the \emph{TEE-based alerting} protocol (illustrated in Figure~\ref{fig:commit_reveal}).
This protocol works under the standard bounded-delay network synchrony model, where the time to write a transaction to the blockchain is not constant but bounded by~$\NetworkWrite$ steps.
We start with a high-level overview of the protocol, where we present its main components, including a Proof of Publication (PoP) primitive~\cite{Ekiden}.
Finally, we present the protocol in detail and show that it implements the simultaneous game.

\begin{figure}[t]
\centering
\includegraphics[width=0.48\textwidth]{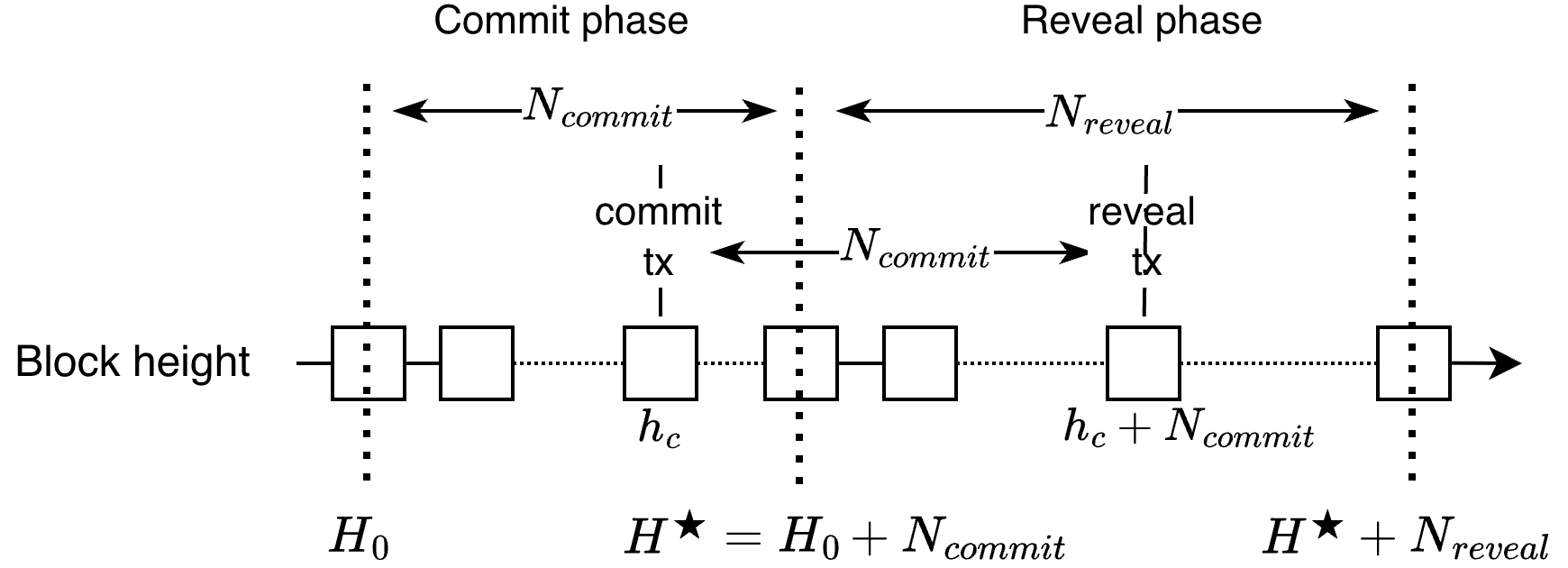}
\caption{TEE-based alerting protocol: all nodes commit during the commitment phase, then reveal after~$\CommitBlocks$ blocks. }
\label{fig:commit_reveal}
\end{figure}

\subsubsection{Motivation and Overview}
First, notice that if we were to use the Lockstep alerting protocol, a node could send its alert transaction after the alerting period ends and, with some probability, still have it included in the blockchain in time for the alerting round. 
This means that if a node sends its action early, other nodes can observe evidence of it and choose their actions accordingly, thereby breaking the information structure of the simultaneous alerting game.
Conversely, if a node delays and sends its action late, less than~$\NetworkWrite$ before the smart contract checks for alerts, it risks the transaction being excluded from the alerting block, violating the action structure of the game (as this node might not be able to choose the~$\alert$ action).
A naive solution to this problem is to have each node encrypt its action before sending it, and reveal the decryption key after the alerting period ends.
While this approach ensures a node can always alert if it chooses to, it fails to achieve Deniability. 
An adversary can  bribe a node to reveal the decryption key early, allowing the adversary to know that the node chose a specific action before the alerting period ends, resulting in an attack that costs a sub-quadratic bribe.
In Appendix~\ref{app:early_reveal_attacks}, we give an example of such an attack and show that bribery cost is~$\Theta(n\cdot \log n)$. 

To prevent such early proofs of action we design a timed commitment scheme that uses \emph{trusted execution environments} (TEEs) and a \emph{Proof of Publication~(PoP)} primitive~\cite{Ekiden} to ensure that nodes cannot reveal their actions before a minimal amount of time elapsed.

Each node has a trusted execution environment: a hardware-backed execution environment that runs a program~$T$ while protecting it from a potentially malicious host. 
Concretely, it provides (1)~\emph{confidentiality} for~$T$'s code and data during execution, so the host cannot read ~$T$'s memory, 
(2)~\emph{integrity} for~$T$'s code and data, so the host cannot tamper with~$T$'s memory, 
(3)~\emph{sealed storage}, which lets~$T$ persist secrets (e.g., nonces) across reboots, and
(4)~\emph{remote attestation}, where the device produces and signs an \emph{attestation} using a device-specific private key whose public key is certified by the hardware vendor, allowing others to verify which program is running in the TEE. 
Thus, for a program~$T$ running inside a TEE, the host cannot observe or modify~$T$'s internal state, and the only publicly verifiable outputs are \emph{attestations} that~$T$ explicitly produces and that others can use to verify~$T$'s execution.
Only the node can interact with its TEE, and the adversary cannot observe or interfere with this interaction.
Critically, the TEE cannot directly access any external information, it relies on the node to supply data for verification, which the TEE then validates cryptographically.

The TEE alerting protocol divides the alerting period into two phases, a \emph{commitment phase} and a \emph{reveal phase}. 
We define the length of these phases in terms of blocks.
The commitment phase lasts for~$\CommitBlocks$ blocks, and the reveal phase lasts for~$\RevealBlocks$ blocks, where~${\RevealBlocks > \CommitBlocks}$.
Both phases are longer than the maximum network delay, i.e., assuming a block is added to the blockchain every~$\blockInterval$ steps, we have~${\CommitBlocks, \RevealBlocks \geq \lceil \NetworkWrite / \blockInterval \rceil}$.
Let~$H_0$ be the block height at the start of the alerting period, and define a \emph{reveal barrier} to be~${H^\star := H_0 + \CommitBlocks}$.
Nodes submit their commitments during the commitment phase, and reveal their actions~$\CommitBlocks$ blocks after they commit, during the reveal phase.

\smallskip\noindent{\textbf{{Blockchain Structure and Block Headers.}}}
A blockchain is a chain of blocks, each containing transactions and a \emph{header} that cryptographically commits to both the block's transactions and the previous block's header.
Additionally, this structure allows to verify the presence of any specific transaction at a specific point in the blockchain.
The TEE uses this structure to verify block progression and that a transaction was published at a particular block height.

\smallskip\noindent{\textbf{{Proof of Publication (PoP) primitive.}}}
Upon first execution, the TEE program~$T$ loads a hardcoded trusted checkpoint, for example, the first block on the same chain as \AlertsCon. 
The TEE updates this checkpoint as the node proves the addition of new blocks.
This ensures that all subsequent block verifications are cryptographically tied to the correct blockchain where the checkpoint is.
To verify that sufficient time has elapsed, the TEE validates a chain of {block headers} input by the node, where each header commits to its predecessor.
This allows us to instantiate a \emph{proof-of-publication (PoP)} protocol~\cite{Ekiden}, to prove a transaction was published before the end of the commitment phase.
Specifically, to verify that at least~$k$ blocks have passed since a transaction~$\tx$ reached the blockchain, before the node sends the transaction, the TEE stores a checkpoint header~$C_B$ and generates a random challenge~$r$.
The node creates the transaction~$\tx$ that is cryptographically tied to~$r$ (e.g., includes a hash of~$r$) and sends the transaction~$\tx$ to the blockchain.
To prove that at least~$k$ blocks have elapsed since~$\tx$, the node sends to the TEE a sequence of block headers that starts from the checkpoint~$C_B$, and includes all the block headers up to the block that includes~$\tx$ and the next~$k$ headers.
I.e., it sends a subchain~${\Pi_\textit{PoP}=(C_B,B_1,\dots,B_l)}$, in which the transaction~$\tx$ tied to~$r$ appears in block~$B_i$ and~$l-i \ge k$, which the TEE verifies.
This provides a TEE-verifiable guarantee that at least~$k$ blocks have elapsed~since~$\tx$. 

{
After successfully verifying a PoP whose final header is~$B_l$, the TEE updates its internal checkpoint to~$B_l$.
The TEE maintains this checkpoint monotonically.
This prevents rollback across protocol rounds and ensures that all PoP verifications are performed with respect to a non-decreasing view of the blockchain.
}

{Crucially, the PoP must be bound to the same chain that hosts $\AlertsCon$.
An adversary could otherwise feed the TEE~$k$ headers from some other chain (e.g., a faster side chain) to claim that the deadline has passed, which says nothing about progress on the blockchain where \AlertsCon is deployed.
To prevent this, the TEE accepts only PoP instances tied to the same chain as \AlertsCon: it verifies the chain's identity via the checkpoint block hash and checks that the commitment transaction~$\tx$ appears on that chain.
This ties the notion of elapsed time to the same blockchain that enforces rewards and penalties.}

{We assume that no single node nor the adversary can fork the blockchain after a sufficiently recent checkpoint block~$C_B$.
In proof-of-work blockchains, creating a fork requires significant computational resources that individual nodes lack.
In proof-of-stake blockchains, creating a fork would result in slashing, which rational nodes avoid.}

In our case, nodes use the TEE to commit to their chosen action.
We use PoP to ensure that a node can only reveal its committed action after~$\CommitBlocks$ blocks have been added to the blockchain since it published its commitment.
Since the reveal period is longer than the commit period~($\RevealBlocks > \CommitBlocks$), all nodes have sufficient time to reveal their actions after the commitment phase ends.


\subsubsection{Participants and Interfaces}
Each node runs a program~$T$ inside a TEE. 
{We denote by~$\mathsf{Hash}$ a cryptographic hash function.}
The program maintains a trusted checkpoint of the blockchain (e.g., the hash of the first block) and exposes the following API:
\begin{itemize}
    \item \textsf{Seal}$(m, \CommitBlocks)$: 
    Takes a message~$m$ (the action and, if an alert is needed, a hash of a proof that it is, otherwise,~$\bot$) and the required delay in blocks~$\CommitBlocks$. 
    It generates a random nonce~$r$, and stores the tuple $(m, r, \CommitBlocks)$ in its sealed storage.
    It returns a public commitment handle~${\gamma=\mathsf{Hash}(m,r,\CommitBlocks)}$.    
    \item \textsf{Unseal-After}$(\gamma, \Pi_\textit{PoP})$: 
    Takes the commitment handle~$\gamma$ and a Proof of Publication~$\Pi_\textit{PoP}$ provided by the host. 
    The proof~$\Pi_\textit{PoP}$ contains a sequence of headers and an inclusion proof of the commitment transaction.
    The TEE verifies that:
    (1)~The headers form a valid chain extending from its trusted checkpoint;
    (2)~The commitment~$\gamma$ is included in a block within this chain; and
    (3)~At least~$\CommitBlocks$ blocks have been appended after the inclusion block.
    If valid, it returns the secret opening~$(m,r)$ and a signed attestation~$A$, {then updates its checkpoint to the last header in~$\Pi_\textit{PoP}$}.
\end{itemize}

  The smart contract $\AlertsCon$ records commitments in an array \textit{commit}, validates timed reveals, and stores valid reveals in an array \textit{reveal}.
  Both arrays are indexed by node identifiers and initialized to~$\bot$.
  It also verifies the TEEs' remote attestations~\cite{chen2025agora, zhang2024teamwork, socrates1024_2023_demystifying, marlin2024nitro} to ensure the node performed the time-check.

\subsubsection{Protocol Description}
  
We detail the commitment and reveal phases of the TEE-based alerting protocol.
The full protocol details can be found in Appendix~\ref{app:tee_based_protocol}.

\textbf{Commitment Phase: }
Each node chooses an action~${a \in \{\alert, \noalert\}}$.
If $a=\alert$, it generates a validity proof~$\Pi_\alert$ and computes its hash $h_\alert = \mathsf{Hash}(\Pi_\alert)$. 
If ${a=\noalert}$, it sets ${h_\alert = \bot}$.
The node calls~${\textsf{Seal}(m= a \mathbin\Vert h_\alert, \CommitBlocks)}$ to generate the commitment handle~$\committment$ and sends a transaction containing~$\committment$ to the blockchain.
Since the commitment phase length is greater than~$\NetworkWrite$, all nodes can ensure their commitments are included.

\textbf{Reveal Phase: }
After the commitment phase concludes, and after~$\CommitBlocks$ blocks were added to the blockchain after the commitment transaction, the node constructs a Proof of Publication~$\Pi_\textit{PoP}$ consisting of the block headers added to the chain since the last checkpoint that the TEE maintains.
It calls~$\textsf{Unseal-After}(\committment, \Pi_\textit{PoP})$ locally in the TEE.
The TEE verifies the proof and, if sufficiently many blocks have been added to the blockchain, returns the tuple~$(m,r)$ containing the message and the nonce and an attestation~$A$.
The node then sends a reveal transaction containing the action~$a$, the proof~$\Pi_\alert$, the nonce~$r$, and the attestation~$A$.

\textbf{On-Chain Verification: }
After the reveal phase ends, the smart contract~$\AlertsCon$ verifies the following for the input from each node~$i$:
First, that the commitment~$\committment_i$ was published during the commitment window, i.e., after height~$H_0$ and before height~$H^\star$.
Then that the attestation~$A$ is valid and signed by a trusted TEE (e.g., as in~\cite{chen2025agora, zhang2024teamwork, socrates1024_2023_demystifying, marlin2024nitro}).
Finally, that the revealed values $(a, h_\alert, r)$ match the commitment~$\committment_i$, and
if $a=\alert$, the hash of the full proof~$\Pi_\alert$ matches~$h_\alert$ and that the proof~$\Pi_\alert$ itself is valid.

\textbf{Rewards, Penalties, and Triggering Alerts: }
If at least one valid alert is revealed, the protocol triggers an on-chain alert, penalizes all non-alerting nodes with penalty~$\penaltySlash$, and rewards all alerting nodes by sharing the slashed penalty symmetrically.
If a node does not send a commitment during the commitment phase, does not reveal its action during the reveal phase, or sends an invalid reveal, the node is considered a non-alerter.
In this case, the protocol triggers an alert on-chain regardless of other nodes' actions.

\begin{note2}[Missing and Incorrect Commitments/Reveals]
  \label{note:incorrect_commit_reveal}
If a node sends an incorrect commitment or reveal transaction, or does not send one at all, it triggers an alert on-chain regardless of other nodes' actions.
An adversary whose goal is to prevent alerts will not benefit from causing nodes to miss or incorrectly send their commitments or reveals and thus will not bribe nodes to do so.  
Instead it may bribe nodes to choose the~$\noalert$ action when an alert is needed.
\end{note2}

\subsubsection{Implementing the simultaneous alerting game}
We show that the TEE-based protocol implements the  simultaneous alerting game.
We start by mapping the participants in the protocol to the players in the game.
The set of nodes in the protocol, each with its own TEE, corresponds to the set of nodes in the game and the adversary in the protocol corresponds to the adversary in the game.
We then map the actions of nodes in the protocol to actions in the game and show that nodes can always commit and reveal within the respective phases.
Then, we show that nodes' actions remain private until the reveal phase ends, and that no transferable evidence of action choice exists before the reveal phase ends, proving that the protocol preserves the information structure of the game.
Finally, we show that the payoffs in the protocol match those in the game.

For each node~$i\in\mathcal{N}$, we define its \emph{induced action}~${a_i \in \{\alert,\noalert\}}$  in an execution of the TEE-based alerting protocol
as follows.
If a node follows the protocol and sends an alert, its induced action is~$\alert$; otherwise, its induced action is~$\noalert$.

To show that the two actions above are feasible, we prove that nodes can always commit and reveal within the respective phases.
\begin{lemma}
\label{lem:timing-feasibility}
\lemmaTEETimingFeasibility
\end{lemma}

The proof relies on the length of the commitment and reveal phases being longer than the maximum network delay~$\NetworkWrite$, and the reveal phase being longer than the commitment phase.
The proof is given in Appendix~\ref{app:tee_alert_capability}.

We now turn to show that the TEE-based alerting protocol preserves the information structure of the simultaneous alerting game.

\begin{lemma}
\label{lem:action-privacy-non-verifiability}
The TEE-based alerting protocol achieves Deniability.
\end{lemma}

\begin{proof}
  Consider a node~$i$ and an execution~$E_0$ in which node~$i$ alerts.
Consider another execution~$E_1$ where the adversary and all nodes but~$i$ behave as in~$E_0$, but node~$i$ chooses not to alert.
By definition of the two executions, before node~$i$ commits to its action, both executions are identical and thus indistinguishable. 

Recall that during the commitment phase, node~$i$:
\begin{enumerate}
  \item chooses an action $a_i\in\{0,1\}$,
  \item forms the proof hash~${h_\alert = \mathsf{Hash}(\Pi_\alert)}$, and
  \item invokes the TEE to compute the commitment~${\committment_i}$, the TEE samples a random nonce $r$, outputs a public handle ${\committment_i  = \mathsf{Hash}(m,r,\CommitBlocks)}$, and stores $(m,r)$ in sealed state.
\end{enumerate}
The node then sends a {commit transaction} carrying $\committment_i$ to the contract.
Denote the randomness and the resulting commitments in executions~$E_0$ and~$E_1$ by~$r^0$, $\committment_i^0$ and~$r^1$, $\committment_i^1$, respectively.

Since~$\committment_i^0$ and~$\committment_i^1$ are computed using a cryptographic hash function,
they cannot help distinguish between executions~$E_0$ and~$E_1$.
By the integrity and confidentiality properties of the TEE, the internal state (including the action and randomness) is inaccessible to any party, including the host (node~$i$) in both executions. 
By the unforgeability of TEE attestations, in both executions, the node cannot fabricate a signature claiming to be from $T$.
By the correctness property of program $T$, the only output path for the randomness and thus~the action is the \textsf{Unseal-After} function, which strictly requires a valid PoP for time $t > H^\star$. 
The \textsf{Unseal-After} function takes the commitment handle~$\committment_i$ and a Proof of Publication~$\Pi_\textit{PoP} = (C_B,B_1,\dots,B_m)$, a sequence of headers starting from the checkpoint~$C_B$, where~$\committment_i$ is included in some block~$B_j$ and ~$j+\CommitBlocks \geq m$. 
The TEE verifies that the headers form a valid chain extending from the checkpoint it has, that~$\committment_i$ is included in a block within this chain, and that at least~$\CommitBlocks$ blocks have been appended after the inclusion block.
It returns the secret opening~$(m,r)$ only if all these checks pass.
By our assumption that no single node nor the adversary can fork the blockchain after a sufficiently recent checkpoint block~$C_B$, neither node~$i$ nor the adversary can create a faster path to revealing the node's action.
Therefore, before $H^\star$, no party can obtain the randomness or action from the TEE in either execution.
Thus, in both executions, the only observable difference is the commitment~$\committment_i$, which is indistinguishable between the two executions.

We conclude that no party can distinguish, with more than negligible advantage, between execution~$E_0$ in which node~$i$ chooses $\alert$ and~$E_1$ in which it chooses~$\noalert$ before the commitment phase ends.
Since the contract checks that the commitment transactions are sent only during the commitment phase, and any commitment sent outside this phase is deemed invalid, no party can distinguish between~$E_0$ and~$E_1$ before choosing its action.  
The same argument applies when starting from an execution where node~$i$ does not alert and constructing an execution where it alerts.
\end{proof}

Finally, we prove that the payoffs in the TEE-based alerting protocol match those of the simultaneous alerting game.
\begin{lemma}
\label{lem:payoff-equivalence}
\lemmaTEEPayoffEquivalence
\end{lemma}
The proof follows directly from the protocol description and can be found in Appendix~\ref{app:tee_reward_dist_proof}.

\begin{theorem}
The TEE-based alerting protocol implements the simultaneous alerting game.
\end{theorem}

\begin{proof}
We prove that the game induced by the TEE-based alerting protocol is the simultaneous alerting game.
First, we map the participants in the protocol to the players in the game. 
By Lemma~\ref{lem:timing-feasibility}, the protocol satisfies Alert Capability.
Then, in Lemma~\ref{lem:action-privacy-non-verifiability}, we proved Deniability holds.
Finally, in Lemma~\ref{lem:payoff-equivalence}, we proved that Reward Distribution holds.
By Proposition~\ref{prop:sufficiency}, the TEE-based alerting protocol implements the simultaneous alerting game.
\end{proof}

The TEE-based protocol requires~$2n$ transactions per round.
The protocol runs in constant time, and the cost to bribe all nodes and cause its failure is~$\penaltySlash n (n-1)$ as shown in Lemma~\ref{lem:bribe-lower-bound}.

\section{Sequential Alerting Protocol}
\label{sec:sequential_alerting_protocol}
While the Lockstep and TEE-based protocols both achieve optimal~$O(n^2)$ bribery resistance in constant time, they require linear on-chain storage incurring up to~$n$ or~$2n$ transactions per alerting period. 
When alerts are infrequent or on-chain storage is expensive, this overhead may be undesirable.
We now present a third protocol that trades constant-time execution for reduced storage costs.

The \emph{sequential alerting protocol}~(\S\ref{sec:sequential_alerting_overview}) forces nodes to act in a specific order, one after the other, with full knowledge of all prior actions.
In exchange for~$O(n)$ worst-case time, this protocol incurs \emph{zero} on-chain storage when no alert is needed, making it suitable for scenarios where alerting is rare.
When an alert is raised, the protocol terminates early with a single transaction.
We analyze the game induced by this protocol~(\S\ref{sec:sequential_alerting_game}), characterize its equilibria under different bribe ranges~(\S\ref{sec:seqalertinganalysis}), and show that a pure Nash equilibrium always exists~(\S\ref{sec:spe_analysis}).

\subsection{Protocol Overview}
\label{sec:sequential_alerting_overview}
Instead of enforcing simultaneity, the \emph{sequential alerting} (illustrated in Figure~\ref{fig:sequential_protocol}) breaks the alerting period into~$n$ equal time intervals called \textit{slots}, each of length~$\SlotLen > \NetworkWrite$, and assigns each node a distinct slot in which it can raise an alert.
The protocol assigns slots to nodes according to a permutation~$\pi$.
Denote by~$\pi(s)$ the node in slot~$1 \leq s \leq n$.
And by~$\pi^{-1}$ the inverse of the permutation~$\pi$, such that~$\pi^{-1}(\pi(s)) = s$.
The function~$\pi^{-1}$ maps the node to its assigned slot.
In a slight abuse of notation, we write~$s_i$ instead of~$s_{\pi^{-1}(i)}$ for the slot assigned to node~$i$.

Each node can raise an alert only during its assigned slot.
If a node raises an alert, the smart contract~\AlertsCon verifies the node acted during its assigned slot, and if so, it penalizes all nodes assigned to earlier slots with~$\penaltySlash$ for not raising an alert. 
The alerting node gets the sum of these slashed penalties as a reward.

{
\subsubsection{Protocol Properties}
\label{sec:seq_properties}
The sequential alerting protocol departs from the simultaneous setting by relaxing the constant-time requirement. 
As a result, the original Deniability property no longer holds, but a relaxed version does hold.

\begin{definition}[Relaxed Deniability]\label{req:private_decisions_no_evidence2}
  For any execution~$E_0$ in which node~$i$ alerts (respectively, does not alert), 
  there exists another execution~$E_1$ in which node~$i$ does not alert (respectively, alerts) such that no node~$j$ \emph{in a later slot than~$i$ (i.e.,~$s_j < s_i$)} and no adversary can distinguish between~$E_0$ and~$E_1$ before {choosing their own actions}. 
\end{definition}

Importantly, however, Alert Capability and Reward Distribution are preserved in this protocol, each node can raise an alert during its assigned slot, all non-alerting nodes are penalized with~$\penaltySlash$, and since there can be up to one alerting node, an alerting node gets all the slashed penalties as a reward.

The question of whether other relaxations of these properties result in additional designs of alerting protocols with asymptotically optimal bribery resistance is left to future work.
}
\subsubsection{Node Sequencing}
\label{sec:node_sequencing}
The protocol uses a lexicographic next-permutation algorithm~\cite{Worlton1968TheAO} to sequence the nodes.
Each node and the smart contract maintain  the current order of nodes~$\pi_t$ and have a function~$\nextperm{\cdot}$ that deterministically computes the next permutation~${\pi_{t+1}= \nextperm{\pi_t}}$. 
This operation requires no communication among nodes or with the smart contract.
The full details can be found in Appendix~\ref{app:node_sequencing}.
Our analysis does not depend on the sequencing method.



\begin{figure}[t]
\centering
\includegraphics[width=0.48\textwidth]{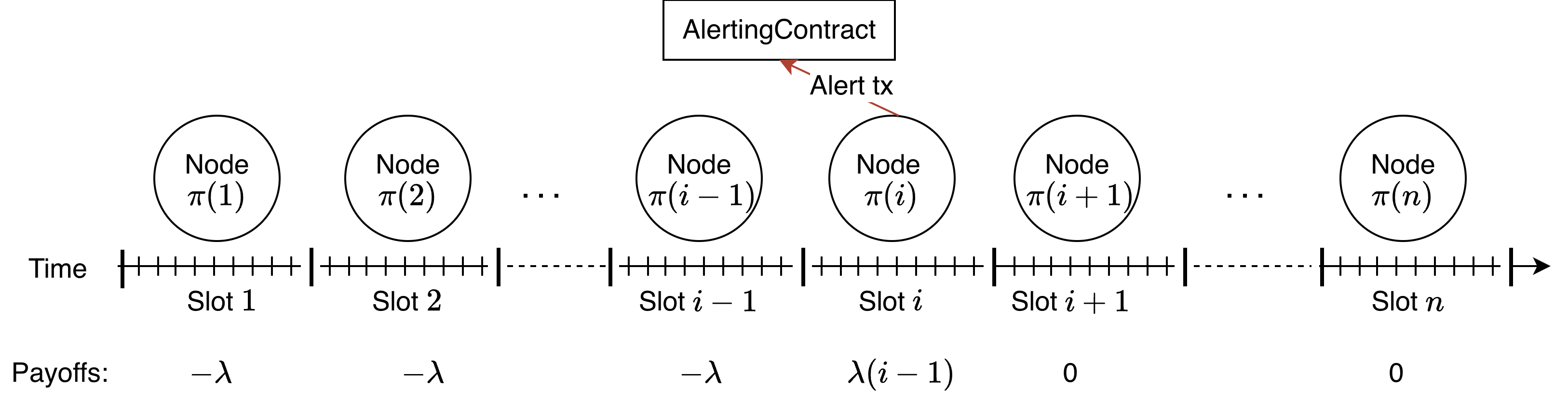}
\caption{Sequential alerting protocol: nodes are assigned to sequential time slots, each node can alert during its slot. The protocol terminates as soon as one node alerts.}
\label{fig:sequential_protocol}
\end{figure}





\subsection{Game Definition}
\label{sec:sequential_alerting_game}
Although the setup is similar to the simultaneous alerting game, in the sequential alerting game, as soon as one node alerts in its assigned slot, the alerting period terminates and later slots are skipped.
Again, two types of players are involved: the set of nodes~$\mathcal{N}$ and an adversary.
The game proceeds across~$n$ time slots. 
In every slot, the node~$i$, assigned to the current slot~$s_i$, can either alert or not alert.
In each time slot, the game consists of two steps:
First, the adversary offers node~$i$ a bribe~$\beta_i$.
Then, node~$i$ decides whether to accept the bribe and skip alerting or to alert and not get the bribe, choosing an action $a_i \in \{\alert, \noalert\}$.
The adversary pays the bribe only to nodes who had the chance to alert and whose on-chain behavior matches the \noalert action.
Denote by~$\beta = (\beta_1, \beta_2, \dots, \beta_n)$ the bribe vector offered by the adversary to all nodes.

The protocol rewards the first node to alert.
In an execution of the game, denote by~$k$ the first node to alert, i.e.,~$a_k = \alert$, and let~$s_k$ be the slot assigned to node~$k$ in the permutation~$\pi$.
Then all other nodes~$j \neq k$ choose $\noalert$, regardless of whether they could have alerted or not.
Node~$k$ that alerts and all nodes~$i$ who did not get the chance to alert, i.e.,~$s_i > s_k$, receive no bribe, we write~$\beta_i = 0$ for all $s_i \geq s_k$.
If no node alerts, we write~$k=\bot$.

Since node~$k$ is the first to alert, all nodes~$i$ such that~$s_i < s_k$ pay~$\penaltySlash$ and take the bribe~$\beta_i$,
all nodes~$i$ such that~$s_i > s_k$ are unaffected, and node~$k$ gets rewarded by the slashed value~$\penaltySlash (s_k-1)$.
If no node alerts, i.e.,~$k=\bot$, then all nodes get their bribe and no one is slashed or rewarded.
To summarize, the payoff of node~$i$ is given by
    \begin{equation}
      \label{eq:seqalerting_payoff_node}
    \rho_i(a_1,\dots,a_n , \beta_i) =
    \begin{cases}
      \penaltySlash (s_k-1), 
        & i = k \\
      \beta_i-\penaltySlash, 
        & s_i < s_k \\
    0,
        & s_i > s_k \\
    \beta_i, 
        & k = \bot \\
    \end{cases}\
    .
  \end{equation}

The adversary pays the bribe only to nodes that get the chance to alert and choose not to.
If no node alerts, the adversary gains~$\advGain$  but pays all bribes.
Therefore, the adversary's payoff is:
$$
\rho_{\text{adv}}(a_1,\dots,a_n, \beta) = 
\begin{cases}
\advGain - \sum_{i=1}^n \beta_i, &  k = \bot\\[6pt]
- \sum_{i=1}^{s_k-1} \beta_i, &  a_k = \alert 
\end{cases}
.
$$

All players' utilities are equal to their payoffs.

\subsection{Equilibrium Analysis}
\label{sec:seqalertinganalysis}

We start by showing that the adversary's gain~$\advGain$ must be greater than~${ \penaltySlash\frac{n(n-1)}{2}}$ to prevent an alert. 

\begin{proposition}
\label{claim:budget_threshold}
If the adversary's gain from an attack~$\advGain$ satisfies~$\advGain \leq  \penaltySlash\frac{n(n - 1)}{2}$, then there exists a node~$i \in \mathcal{N}$ that will alert.
\end{proposition}

\begin{proof}
Consider the adversary's utility.
If she bribes all~$n$ nodes and none alert, then her utility is~$\advGain - \sum_{i=1}^{n} \beta_i$.
If she bribes the first~$k<n$ nodes to not alert, and does not bribe node~$k+1$,
then node~$k+1$ will alert since doing so yields a positive utility of~$\penaltySlash k$ compared to not alerting which will either result in it receiving her bribe~$\beta_{k+1}=0$ or being penalized~$\beta_{k+1} - \penaltySlash =- \penaltySlash < 0$. 
Then the adversary's utility is~$-\sum_{i=1}^{k} \beta_i$.
To maximize her utility, the adversary will bribe nodes only if the gain from an alert not being raised~$\advGain$ is larger than the total bribes paid, i.e., $\advGain > \sum_{i=1}^{n} \beta_i$.

Assume for contradiction that~$\advGain \leq \penaltySlash\frac{n(n-1)}{2}$ and that no node alerts, meaning all nodes choose $\noalert$ in their slot.
Consider an arbitrary node~$i \in \mathcal{N}$ occupying slot~$s_i$. 
From our assumption, all nodes in slots~$s<s_i$ choose $\noalert$. 

If node~$i$ alerts, its utility is~$\penaltySlash (s_i - 1)$, and its utility from not alerting is~$\beta_i$ if all others also choose~$\noalert$, but~$\beta_i - \penaltySlash$ if a node alerts in a later slot. 
Therefore, her utility for not alerting is at most~$\beta_i$.
If~$\penaltySlash(s_i - 1) \geq \beta_i$, node~$i$ will alert.
But we assumed no node alerts, so we must have~$\beta_i > \penaltySlash(s_i - 1)$.
Summing this requirement across all nodes~$i \in \mathcal{N}$, the adversary's total bribe must satisfy:
\[
\sum_{i=1}^{n} \beta_i \geq \sum_{i=1}^{n} \penaltySlash (s_i - 1) = \penaltySlash \frac{n(n-1)}{2}.
\]

However, by assumption, the adversary's gain from the attack is limited to:
$
\advGain \leq \penaltySlash \frac{n(n - 1)}{2}.
$
Therefore, the adversary's total bribe exceeds her gain from the attack, and bribing all nodes to not alert yields a negative utility, lower than her utility if she chooses not to bribe at all.
We conclude that the adversary does not bribe all nodes to not alert.
Thus, at least one node must alert.
\end{proof}

\begin{corollary}
\label{cor:adversary_gain_threshold}
\corProfitableBribery
\end{corollary}
The proof can be found in Appendix~\ref{app:sequential_adversary_optimal}.

\subsection{Subgame-Perfect Nash Equilibrium Analysis}
\label{sec:spe_analysis}
To analyze rational node behavior in the sequential alerting game, we employ the concept of Subgame Perfect Nash Equilibrium (SPNE). 
Since the game is sequential, each node observes previous actions. 
To characterize the SPNE, we apply backward induction starting from the last node in the permutation and working back to the first.
We define the \emph{adversary's remaining budget} in slot~$s$ as the amount of bribes the adversary can still pay to nodes to prevent them from alerting while ensuring her total bribes do not exceed her gain, i.e.,~$B(s) = \advGain - \sum_{j=1}^{s-1} \beta_{\pi(j)}$.

\begin{lemma}[Adversary's Gain Threshold for Alerting]
\label{claim:adversary_gain_threshold_helper}
In an execution of the sequential alerting game, if the nodes assigned to the first~$s-1$ slots did not alert, then the adversary can bribe nodes in slots~$s,\dots,n$ to not alert if and only if the adversary's remaining budget~$B(s)$ is at least~$\sum_{j=s}^{n} \penaltySlash (j-1)$. 
\end{lemma}

\begin{proof}
We use backward induction over the sequentially ordered set of nodes from slot $n$ to slot $1$ to prove that if the nodes assigned to the first~$s-1$ slots did not alert, then the adversary bribes nodes in slots~$s,\dots,n$ to not alert if and only if the adversary's remaining budget~$B(s)$ is at least~$\sum_{j=s}^{n} \penaltySlash (j-1)$. 

\emph{Base Case:}  
Consider the node~$i = \pi(n)$ assigned to the last slot~$n$. 
If any previous node in slots~$1$ to~$n-1$ alerted, then the game would have ended, and node~$i$ would not have the opportunity to alert.
Otherwise, all previous nodes assigned to slots~$1$ to~$n-1$ choose not to alert. 
Thus, by Equation~\ref{eq:seqalerting_payoff_node}, node~$i$'s utility is~$\penaltySlash(n-1)$ if it alerts, $a_i = \alert$,
or~$\beta_{\pi(n)} $ if it does not,~${a_i = \noalert}$.

Hence, node $i$ alerts if and only if~$\penaltySlash(n-1)  \geq \beta_{\pi(n)}$.
So the adversary bribes node $i$ to not alert if and only if the adversary's remaining budget~${B(n) \geq \penaltySlash(n-1)}$.

\emph{Inductive Step:}
Moving backward, consider an arbitrary slot~$s<n$ and assume the claim holds for all~$s<j\leq n$.

I.e., if the nodes assigned to the first~$j$ slots did not alert, the adversary bribes nodes in slots~$j+1,\dots,n$ to not alert if and only if the adversary's remaining budget~$B(j+1)$ is at least~$\sum_{t=j+1}^{n} \penaltySlash \cdot t$. 

If the nodes assigned to the first~$s-1$ slots did not alert, then the node~$i$ assigned to slot~$s$ needs to decide whether to alert or not.
If it alerts, it receives a reward of~$\penaltySlash(s-1)$.
If it does not alert, it receives a bribe of~$\beta_i$, but risks being penalized later if any of the nodes in slots $s+1, \dots, n$ alerts.

By Corollary~\ref{cor:adversary_gain_threshold}, if the adversary cannot bribe all nodes to not alert, then she will not bribe any.
By Equation~\ref{eq:seqalerting_payoff_node}, node~$i$'s utility is~$\penaltySlash(s-1)$ if it alerts,~$a_i = \alert$,
or~$\beta_i$ if it does not alert,~$a_i = \noalert$.
Thus, node~$i$ alerts if and only if $
\penaltySlash(s-1)  \geq \beta_i.
$

The adversary bribes node $i$ and all nodes in later slots~$s+1, \dots, n$ to not alert if and only if the adversary's remaining budget is 
    \[
    B(s) > \penaltySlash(s-1) + B(s+1) \\
    = \sum_{j=s}^{n} \penaltySlash(j-1).
  \]
\end{proof}

\begin{corollary}
\label{cor:sequential_bribery_cost}
  \corBriberyCost
\end{corollary}
We prove by applying Lemma~\ref{claim:adversary_gain_threshold_helper} to all nodes. The proof is found in Appendix~\ref{app:sequential_bribery_cost}

\section{Conclusion}
\label{sec:conclusion}
We formalize the alerting problem in blockchain systems in a cryptoeconomic setting, where~$n$ nodes monitor off-chain events and raise alerts on-chain when certain conditions hold.
We present an upper bound on the bribery resistance that any $\penaltySlash$-alerting protocol can achieve, i.e., the total amount an adversary must pay in bribes to prevent any alert from being published on-chain.
Then, we present a constant-time simultaneous alerting game that asymptotically achieves this upper bound.
We present the Lockstep and the TEE-based alerting protocols that implement this game, each under different assumptions. 
Both protocols are constant-time, but require~$O(n)$ on-chain storage in the worst case.
We relax the constant-time requirement and analyze the sequential alerting protocol that requires~$O(1)$ on-chain storage, at the cost of~$O(n)$ time, while also achieving asymptotically optimal bribery resistance.

%% file: related_work.tex
Many critical workflows in blockchain systems rely on off-chain actors sending alerts with external information to maintain their functionality.
Previous research has examined mechanisms for ensuring that decentralized systems can detect and respond to external events, faults, or attacks, but without modeling the economics of suppressing those detection systems.
Examples of such systems include rollups~\cite{eth-optimistic-rollups,arbitrum-sequencer,optimism-batcher,sok-rollups} execute transactions off-chain and periodically commit snapshots to a base chain.
{Oracle networks}~\cite{chainlink-feeds,uma-oo-overview} deliver external data, such as asset prices or randomness, to smart contracts.
Cross-chain protocols~\cite{cross1,cross2,cross3}, transfer messages and assets between blockchains.
And misbehavior reporting protocols~\cite{cheating1,cheating2,cheating3}, submit evidence of protocol violations on-chain to enforce penalties.

In each of these cases, alert types are different, but they all share the property that a missing or incorrect alert can lead to significant financial loss or protocol failure.
All of these systems rely on off-chain actors to send alerts to trigger on-chain functionality.
The system rewards alerters when they raise alerts based on a pre-defined set of conditions.
Some systems also penalize participants for incorrect or missing alerts (e.g.,~\cite{uma-oo-overview,tellor-disputes}).
However, existing works do not analyze the incentives in the presence of a bribing adversary who aims to suppress alerts.
Meanwhile, other works show that adversaries use bribes to induce protocol deviations in other blockchain systems~\cite{bribes1, bribes2, lloyd2023emergent, vote-buying, bribes5}.

Most existing mechanisms for the blockchain alerting problem pay a fixed bounty to the first valid alert.
However, we show that such mechanisms are vulnerable to bribery-based suppression attacks.
In particular, a bribing adversary can suppress alerts by paying an amount linear in the number of participants, $O(n)$.
Chen et al.~\cite{Prrr} study incentivizing timely on-chain alerts but do not consider a bribing adversary.

The Chainlink~2.0 whitepaper~\cite{chainlink2021whitepaper} proposes a sequential mechanism in which nodes who participate in a protocol to agree on off-chain data, can raise alerts on incorrect outputs of the protocol.
This protocol runs in linear time in the number of nodes.
The whitepaper does not analyze the adversary's incentives or prove the claimed quadratic bound.
We build two constant-time alerting protocols and consider the sequential alerting protocol proposed in~\cite{chainlink2021whitepaper}.
We formalize and analyze the three alerting protocols and compare their trade-offs in terms of bribery resistance, on-chain overhead, latency, and practical assumptions. 
Our work is the first, to our knowledge, to formalize the \emph{alerting problem} as a cryptoeconomic game, analyze its equilibria, establish asymptotically optimal bribery resistance, and specify different practical instantiations.